\documentclass[onecolumn,accepted=2018-08-03]{quantumarticle}
\usepackage[utf8]{inputenc}

\usepackage[english]{babel}

\usepackage{microtype}
\usepackage{xspace}
\usepackage{bbm,amsmath,amssymb,amsthm,url,graphicx,mathtools}
\usepackage{thmtools,thm-restate}
\usepackage{tabstackengine,multirow}
\usepackage[colorlinks,pdftex]{hyperref}

\usepackage[most]{tcolorbox}
\newtcolorbox{proofsomething}[1]{
freelance,
breakable,
colback=white,
colframe=white,
width=\dimexpr\textwidth+28pt\relax,
before=\par\vspace{\bigskipamount}\noindent,
enlarge left by=-14pt,
overlay unbroken and first={
  \node[
  anchor=north east,
  inner xsep=8pt,
  xshift=8pt,
  rounded corners=5pt,
  font=\bfseries,
  fill=white,
  align=right] at ([xshift=-0.2cm]frame.north west) (tit) {\strut Proof of\\\Cref{#1}};
  \draw[
  line width=3pt,
  rounded corners=5pt,
  black!10
  ] ([xshift=4pt]frame.north west) -- ([xshift=4pt]frame.south west);
},
overlay middle and last={
  \draw[
  line width=3pt,
  rounded corners=5pt,
  black!10
  ] ([xshift=4pt]frame.north west) -- ([xshift=4pt]frame.south west);
},
frame code={},
interior code={},
top=0pt,
bottom=0pt
}

\usepackage{tikz}
\usetikzlibrary{matrix}

\usepackage{cleveref}
\crefname{lemma}{lemma}{lemmas}
\crefname{proposition}{proposition}{propositions}
\crefname{definition}{definition}{definitions}
\crefname{theorem}{theorem}{theorems}
\crefname{conjecture}{conjecture}{conjectures}
\crefname{corollary}{corollary}{corollaries}
\crefname{example}{example}{examples}
\crefname{section}{section}{sections}
\crefname{appendix}{appendix}{appendices}
\crefname{figure}{fig.}{figs.}
\crefname{equation}{eq.}{eqs.}
\crefname{table}{table}{tables}
\crefname{item}{property}{properties}
\crefname{remark}{remark}{remarks}

\newtheorem{theorem}{Theorem}
\newtheorem{definition}[theorem]{Definition}
\newtheorem{corollary}[theorem]{Corollary}
\newtheorem{proposition}[theorem]{Proposition}
\newtheorem{lemma}[theorem]{Lemma}

\newtheorem{remark}[theorem]{Remark}
\include{braket}

\DeclareMathOperator{\diag}{diag}
\DeclareMathOperator{\poly}{poly}
\DeclareMathOperator{\spec}{spec}

\DeclareMathOperator{\supp}{supp}
\DeclareMathOperator{\rank}{rank}
\DeclareMathOperator{\tr}{tr}
\DeclareMathOperator{\dist}{dist}
\DeclareMathOperator{\diam}{diam}

\newcommand{\op}{\mathbf}
\newcommand{\abs}[1]{|#1|}

\renewcommand{\H} {{\ensuremath{\mathcal H}}\xspace}
\newcommand{\1} {\ensuremath{\mathbbm 1}}

\newcommand{\dd}{\mathrm{d}}

\newcommand{\field}[1] {\mathbbm{#1}}

\newcommand{\QMA}{\textnormal{{QMA}}\xspace}
\newcommand{\QMAEXP}{\textnormal{{QMA\textsubscript{EXP}}}\xspace}
\newcommand{\yes}{\textnormal{{YES}}\xspace}
\newcommand{\no}{\textnormal{{NO}}\xspace}

\newcommand{\BigO}{O}
\newcommand\smallO{
  \mathchoice
    {{\scriptstyle\BigO}}
    {{\scriptstyle\BigO}}
    {{\scriptscriptstyle\BigO}}
    {\scalebox{.7}{$\scriptscriptstyle\BigO$}}
  }
\DeclareMathOperator{\lmin}{\lambda_\mathrm{min}}
\DeclareMathOperator{\lnz}{\lambda_{\neq0}}

\makeatletter
\newcommand{\subalign}[1]{%
  \vcenter{%
    \Let@ \restore@math@cr \default@tag
    \baselineskip\fontdimen10 \scriptfont\tw@
    \advance\baselineskip\fontdimen12 \scriptfont\tw@
    \lineskip\thr@@\fontdimen8 \scriptfont\thr@@
    \lineskiplimit\lineskip
    \ialign{\hfil$\m@th\scriptstyle##$&$\m@th\scriptstyle{}##$\crcr
      #1\crcr
    }%
  }
}
\makeatother
\newcommand{\Hprop}{\op H_\textnormal{prop}}
\newcommand{\Hclock}{\op H_\textnormal{clock}}

\newcommand{\HFK}{\op H_\textnormal{FK}}
\newcommand{\Hkit}{\op H_\textnormal{Kitaev}}
\newcommand{\Hin}{\op H_\textnormal{in}}
\newcommand{\Hout}{\op H_\textnormal{out}}
\newcommand{\Pin}{\Pi_\textnormal{in}}
\newcommand{\Pout}{\Pi_\textnormal{out}}

\usepackage[backend=bibtex,style=numeric-comp,doi=false,isbn=true,url=false,maxbibnames=20,sorting=none]{biblatex}
\bibliography{bibliography}
\newbibmacro{string+doi}[1]{\iffieldundef{doi}{#1}{\href{https://doi.org/\thefield{doi}}{#1}}}
\DeclareFieldFormat{title}{\usebibmacro{string+doi}{\mkbibemph{#1}}}
\DeclareFieldFormat[article,thesis,incollection,inproceedings]{title}{\usebibmacro{string+doi}{\mkbibquote{#1}}}

\iffalse
\usepackage{showlabels}
\usepackage[colorinlistoftodos]{todonotes}
\usepackage{soul}

\else
\newcommand{\listoftodos}{}
\fi

\begin{document}
\thispagestyle{empty}
\listoftodos
\newpage

\title{Analysis and limitations of modified\\circuit-to-Hamiltonian constructions}
\author{Johannes Bausch}
\email{jkrb2@cam.ac.uk}
\affiliation{DAMTP, CQIF, University of Cambridge}
\orcid{0000-0003-3189-9162}
\author{Elizabeth Crosson}
\email{crosson@caltech.edu}
\affiliation{IQIM, California Institute of Technology}
\orcid{0000-0002-7847-6354}


\begin{abstract}
Feynman's circuit-to-Hamiltonian construction connects quantum computation and ground states of many-body quantum systems.  Kitaev applied this construction to demonstrate \QMA-completeness of the local Hamiltonian problem, and Aharanov et al.\ used it to show the equivalence of adiabatic computation and the quantum circuit model.   In this work, we analyze the low energy properties of a class of modified circuit Hamiltonians, which include features like complex weights and branching transitions.  

For history states with linear clocks and complex weights, we develop a method for modifying the circuit propagation Hamiltonian to implement any desired distribution over the time steps of the circuit in a frustration-free ground state, and show that this can be used to obtain a constant output probability for universal adiabatic computation while retaining the $\Omega(T^{-2})$ scaling of the spectral gap, and without any additional overhead in terms of numbers of qubits.

Furthermore, we establish limits on the increase in the ground energy due to input and output penalty terms for modified tridiagonal clocks with non-uniform distributions on the time steps by proving a tight $\BigO(T^{-2})$ upper bound on the product of the spectral gap and ground state overlap with the endpoints of the computation.
Using variational techniques which go beyond the $\Omega(T^{-3})$ scaling that follows from the usual geometrical lemma, we prove that the standard Feynman-Kitaev Hamiltonian already saturates this bound.

We review the  formalism of unitary labeled graphs which replace the usual linear clock by graphs that allow branching and loops, and we extend the $\BigO(T^{-2})$ bound from linear clocks to this more general setting.
In order to achieve this, we apply Chebyshev polynomials to generalize an upper bound on the spectral gap in terms of the graph diameter to the context of arbitrary Hermitian matrices.
\end{abstract}
\maketitle
\newpage

\tableofcontents

\section{Introduction}
The initial demonstration of QMA-completeness of the local Hamiltonian problem~\cite{Kitaev2002} was followed by a period of development during which the main goal was to broaden the class of interaction terms which suffice to make the local Hamiltonian problem QMA-complete~\cite{Kempe2006,Oliveira2008}.
These results were motivated in part by a desire to understand the hardness of approximating physical systems that resemble those found in nature, and also by the goal of making the closely related universal adiabatic computation construction~\cite{aharonov2008adiabatic} more suitable for eventual physical implementation~\cite{biamonte2008realizable}.
The success of these efforts have resulted in QMA-complete local Hamiltonian problems with restricted properties such as 2-local interactions~\cite{Kempe2006}, low dimensional geometric lattices~\cite{Oliveira2008,Aharonov2009}, and translational invariance, as well as a complete classification of the complexity of the 2-local Hamiltonian problem for any set of interaction couplings~\cite{cubitt2016complexity}.

This great success in classifying the hardness of physically realistic interactions stands in contrast with the relative lack of progress in resolving questions related to the robustness of quantum ground state computation, such as whether fault-tolerant universal adiabatic computing is possible, or to prove or disprove the quantum PCP conjecture~\cite{aharonov2013guest}.
Such questions motivate us to seek (or to limit the possibility of) improvements to the circuit-to-Hamiltonian construction itself, which serves a foundational role in all of the results listed above.
Based on ideas by Feynman~\cite{Feynman1986} and cast into its current form by Kitaev~\cite{Kitaev2002}, the construction remains relatively little changed, having undergone only some gradual evolutions since its modern introduction~\cite{Gosset2013,Aharonov2009,Gottesman2009,breuckmann2014space,Hallgren2013,NagajTalk,Bausch2016}.
Only recently steps have been undertaken to analyse modifications in depth \cite{Usher2017}, and revisit the spectral properties of Kitaev's original construction \cite{Caha2017}.

To encode computation robustly the circuit-to-Hamiltonian construction should not only have a ground space representing valid computations, but intuitively it should penalize invalid computations with as high of an energy as possible.
One way of formalizing this condition is to add constraints on the input and the output of the circuit that cannot be simultaneously satisfied under the valid operation of the circuit.
If the Hamiltonian enforces the correct operation of the circuit gates, then the input and output constraints that contradict each other should not be satisfiable by any state, and so the ground state energy should increase.
This explains why the higher ground state energy associated with non-accepting circuits can be regarded as an energy penalty against invalid computations, which we call the ``quantum UNSAT penalty''.

The specific role of the $\Omega(T^{-3})$ scaling of the quantum UNSAT penalty in Kitaev's proof is to show that the local Hamiltonian problem is QMA-hard with a promise gap that scales inverse polynomially in the system size.
While there exists a well-defined relation between runtime $T$ and the corresponding Hamiltonian's system size $n$ for any \emph{specific} set of constructions---e.g.\  Kitaev's 5-local one---the explicit scaling of this gap with $T$ is not meaningful to the local Hamiltonian problem beyond the fact that it is polynomial, since the local Hamiltonian problem promise gap is parameterized by the number of qubits $n$, i.e.~$1/\poly n$.

Nevertheless, the scaling of the quantum UNSAT penalty with $T$ is a well defined feature of any particular circuit-to-Hamiltonian construction, and therefore we take the view that it is a reasonable metric for exploring the space of possible improvements to this construction.
The goal of this paper is to answer a series of specific questions revolving around the UNSAT penalty.
\begin{enumerate}
\item Is the geometrical lemma tight for Kitaev's original construction?
\item Can the circuit-to-Hamiltonian construction be improved with regard to universal adiabatic computation?
\item Can we modify the Feynman-Kitaev construction---by introducing clock transitions with varying weights, or by adding branching or loops as analysed in the context of unitary labeled graphs---in order to improve on the known $\Omega(T^{-3})$ bound on the UNSAT penalty?
\end{enumerate}
In the next section, we motivate each of these questions and formally state our results.

\section{Results and Overview}
\subsection{Background and Notation}
We analyze history state Hamiltonians on a bipartite Hilbert space $\H=\field C^{T+1}\otimes \H_q$, which are called clock and quantum register, respectively.
In the bulk of this paper, we consider history states consisting of an arbitrary complex superposition of the time steps of a computation,
\begin{equation}
	\ket\psi = \sum_{t = 0}^T \psi_t \ket t \otimes \left (\op U_t \cdots \op U_1\right) \ket\xi, \label{eq:weighted}
\end{equation}
where as usual $\op U_T,\ldots,\op U_1$ are (local) quantum gates, $\ket\xi$ is an arbitrary input to the computation, and $\ket\psi$ is a normalized state, so that in particular $\pi_t := \psi^*_t \psi_t$ is a probability distribution over $\{0,\ldots,T\}$.
Hamiltonians with ground states of the form as in \cref{eq:weighted} arise from modifications to the usual terms of the Feynman circuit Hamiltonian,
\begin{equation}\label{eq:ham}
	\Hprop := \sum_{t = 0}^T a_t \ketbra t \otimes \1 + \sum_{t =0}^{T-1} \left( b_t\ketbra{t+1}{t}\otimes \op U_t + b^*_t \ketbra{t}{t+1}\otimes \op U^\dagger_t \right),
\end{equation}
where we restrict $|a_t|,|b_t|\leq 1$ for $t = 0,\ldots,T$.
Note that most if not all of the constructions that break down the $\BigO(\log n)$-local time register interactions to a $k$-local Hamiltonian--- such as the domain wall clock which leads to a 5-local circuit Hamiltonian---can be directly applied to the modified form \cref{eq:ham}, up to a constant prefactor (see \cref{eq:standardHam}).
In \cref{lem:clockham}, we show that $\Hprop$ in \cref{eq:ham} is unitarily equivalent to a Hamiltonian $\Hclock\otimes\1_q$---i.e.\ acting trivially on the quantum register---where
\begin{equation}\label{eq:clockham}
    \Hclock = \sum_{t=0}^T a_t \ketbra t\otimes\1 + \sum_{t=0}^{T-1}(b_t\ketbra{t+1}{t} + b_t^*\ketbra{t}{t+1}).
\end{equation}

In addition to the part of the Hamiltonian which checks the propagation of the circuit, projectors $\Hin:= \ketbra 0\otimes \Pin$ and $\Hout := \ketbra T\otimes \Pout$ can be added to $\Hprop$ to validate specific inputs and the outputs of the computation.
We define the sum of all these terms to be the \emph{modified Feynman-Kitaev Hamiltonian},
\begin{equation}\label{eq:mFK}
	\HFK:=\Hprop + \Hin + \Hout,
\end{equation}
as opposed to the Feynman circuit Hamiltonian in \cref{eq:ham}, which is just the propagation part $\Hprop$ without the penalties.

More specifically, the ground space of $\Hprop + \Hin$ will be spanned by computations starting from a valid input computation (i.e.~those for which $\ket\xi\in\ker\Pin$ in \cref{eq:weighted}), and $\Hout$ will raise the energy of the state $\ket\psi$ in \cref{eq:weighted} when $\op U_T\cdots \op U_1\ket\xi\not\in\ker\Pout$.
The magnitude of this frustration between the incompatible ground spaces of $\Hin + \Hprop$ and $\Hout$ will depend on the circuit encoded by $\HFK$ and the specific in- and output energy penalties, i.e.~the maximum acceptance probability of the circuit
\begin{equation*}
	\epsilon:=\max_{\subalign{\ket\xi &\in \ker \Pin\\\ket\eta &\in \ker \Pout} }=\bra\eta \op U_T\cdots \op U_1\ket\xi.
\end{equation*}
In the following definition we take the $\Pin$ and $\Pout$ that are used in the standard construction: $\Pout :=\proj 0$ measures a single qubit and penalizes it in state $\ket 0$ (i.e.\ ``not accepted''), and $\Pin$ constrains a fraction of the input qubits to the $\ket 0$ state as ancillas, some to an encoded string describing the problem instance, and leaves the rest of the input qubits unconstrained to act as a witness, which allows the embedded computation to act as a verifier circuit for e.g.\ \QMA or \QMAEXP verification purposes.

For a specific set of fixed in- and output constraints and fixed runtime length $T$, we want to identify the circuit Hamiltonian best suitable to discriminate accepting and rejecting circuit paths, independent of the particular circuits used.
Let $\epsilon>0$ be some acceptance probability.
We define the quantity $C(\epsilon,T)$ to be the set of those circuits of size $T$, for which the maximum acceptance probability amongst any states obeying the in- and output constraints $\Pin$ and $\Pout$ is precisely $\epsilon$:
\begin{equation*}
	C(\epsilon,T) := \{\op U_1,\ldots,\op U_T :  \max_{\subalign{\ket\xi &\in \ker \Pin\\\ket\eta &\in \ker \Pout} } |\bra\xi \op U_1 \cdots \op U_T \ket\eta|^2 = \epsilon\}.
\end{equation*}
Note that we allow the $\op U_i$ to be completely arbitrary---the precise action of the circuit, whether it performs Grover search or is just a random set of gates, is irrelevant for $C(\epsilon,T)$ as we have defined it here.
This makes sense: we want to completely decouple the notion of how well a modified circuit Hamiltonian can penalize in- and output from what the circuit itself does.

This leads to our definition of the quantum UNSAT penalty of a circuit-to-Hamiltonian construction, which quantifies the extent to which a Hamiltonian as in \cref{eq:ham} can enforce the input and output penalties described above for an arbitrary circuit.
\begin{definition}\label{def:unsat}
	Let the $E(\HFK)$ and $E(\Hprop)$ be the ground state energies of $\HFK$ and $\Hprop$, respectively, and define the quantum UNSAT penalty $E_p(\epsilon,T)$
	\[
		E_p(\epsilon, T) := \min_{\op U_1,\ldots,\op U_T \in C(\epsilon,T) } E(\HFK) - E(\Hprop).
	\]
\end{definition}

The reason for explicitly subtracting the energy term $E(\Hprop)$ is that this quantity is not necessarily zero for propagation Hamiltonians of the form \eqref{eq:ham}, and so in these cases of frustrated propagation terms the UNSAT penalty is defined to be the additional energy due to the rejection probability of the circuit.\footnote{The term ``UNSAT'' derives from the related QUNSAT$_\psi=\sum_{e\in E}\bra\psi \op h_e \ket\psi / |E|$ quantity that was previously defined and analyzed using the detectability lemma \cite{Aharonov2009a}. 
We use the term UNSAT \emph{penalty} to emphasize that it is the energy difference between accepting and non-accepting computations.}

\subsection{Results}
\subsubsection{Non-Uniform Circuit Histories}
Our first step in analyzing the UNSAT penalty of modified Feynman Hamiltonians is to apply  a modified version of the argument used in the standard construction (see \cite[sec.~14.4]{Kitaev2002}) to ``undo'' the computation and show that $\Hprop$ is unitarily equivalent to a clock Hamiltonian which acts trivially on the computational register.

\begin{restatable}{lemma}{clockham}
	\label{lem:clockham}
	If $\op W := \sum_{t=0}^T \proj t \otimes (\op U_t \cdots \op U_1)$, then $\op W$ is unitary and $\op W^\dagger \Hprop \op W = \op H_\textnormal{clock} \otimes \1$, where the clock Hamiltonian $\Hclock$ is given by \cref{eq:clockham}.
\end{restatable}

Next we apply the same geometrical lemma used in Kitaev's proof to lower bound the UNSAT penalty of modified Feynman Hamiltonians.
\begin{restatable}{lemma}{geo}
	\label{lem:geo} 
	If the spectral gap of the corresponding clock Hamiltonian $\Hclock(T)$---denoted  $\Delta_{\op H}(T)$---is less than the (constant) spectral gap of $\Hin + \Hout$, then
	\begin{equation}
		\label{eq:geo} \frac{\Delta_{\op H}(T)}{4} (1 - \sqrt{\epsilon}) \times \min\{\pi_0, \pi_T\} \le E_p(\epsilon,T) \le E\left(\Hclock(T) + \ketbra0 + \ketbra T\right ) .
	\end{equation}
\end{restatable}

The upper bound follows from the operator inequalities $\Pi_\textrm{in} \preceq \mathbbm{1}$ and $\Pi_\textrm{out} \preceq \mathbbm{1}$, and it says that the increase in the ground state energy due to the penalty terms is bounded by the case when all of the frustration is in the system's time register.
The lower bound states that the UNSAT penalty can be increased either by boosting the spectral gap of the clock Hamiltonian, or by amplifying the overlap of the ground state with the beginning and ending time steps of the computation (i.e.\ modifying the ground state $\ket{\psi_0}$ of $\Hprop$ such that $\bra{\psi_0}\Pin\ket{\psi_0}$ is maximized, and analogously for $\Pout$).
To see that the overlap with the endpoints of the computation can be made arbitrarily close to one, we prove that it is in fact possible to construct a clock Hamiltonian with an arbitrary distribution as its ground state.

\begin{restatable}{lemma}{mcham}
	\label{lem:mcham}
	For any probability distribution $\mu$ with support everywhere on its domain $\{0,\ldots,T\}$, there is a choice of coefficients $\{a_t,b_t\}_{t=0}^T$ in \cref{eq:clockham} such that $\Hclock$ is stoquastic and frustration-free  and has a ground space spanned by states of the form \cref{eq:weighted} with weights $\psi_t = \sqrt{\mu_t}$.
\end{restatable}

Using \cref{lem:mcham} we exhibit a modified Hamiltonian for a target ground state distribution with $\pi_0,\pi_T\ge1/4$, and show that it has a spectral gap that is $\Omega(T^{-2})$ to establish our first main result; in order to prove it, we use the \emph{geometrical lemma} for modified Feynman Hamiltonians, a by now standard but central proof technique that is used ubiquitously in hardness constructions of the local Hamiltonian problem.
\begin{restatable}{theorem}{main}
	\label{theo:main}
	There is a frustration-free modified circuit Hamiltonian as in \cref{eq:ham} with an UNSAT penalty $E_p(\epsilon, T)$ that is $\Omega((1 - \sqrt{\epsilon})T^{-2})$.
\end{restatable}

In \cref{sec:padding-ham} we use the geometrical lemma to show that padding the standard circuit-to-Hamiltonian construction with $T$ input and output penalties also achieves an $\Omega(T^{-2})$ scaling of the UNSAT penalty.

\subsubsection{Universal Adiabatic Computation}
The modified circuit Hamiltonian used in \cref{theo:main} also has a potential advantage for universal adiabatic computation (AQC).  The original construction of universal AQC linearly interpolates between the initial $\op H_{\textrm{init}}:= \op H_{\textrm{in}}$ constraint and the final Hamiltonian $\op H_{\textrm{final}} := \op H_{\textrm{in}} + \op H_{\textrm{prop}}$,
\begin{equation}
\op H(s) = (1-s) \op H_{\textrm{init}} + s \op H_{\textrm{final}}.\label{eq:adiabaticLinear}
\end{equation}
By initializing the system in the ground state of $\op H_{\textrm{init}}$ at s = 0, and increasing $s$ to 1 in a continuous fashion and sufficiently slowly, the system will remain near its instantaneous ground state.
This requires the total time for $s$ to rise from 0 to 1 to scale as $\textrm{poly}(n,\Delta_{\min}^{-1})$, where $\Delta_{\min} := \min_s \Delta_{\op H(s)}$ and $n$ the system's size.  In~\cite{aharonov2008adiabatic} it is proven that $\Delta_{\min} = \Omega(T^{-2})$ for the linear interpolation \eqref{eq:adiabaticLinear} and the standard $\op H_{\textrm{prop}}$ from Kitaev's construction.

Since the goal of universal AQC is to efficiently produce the state at the final adiabatic time $s = 1$, one needs to measure the clock register and condition success on finding the computation in its final  clock state, i.e.\ $t = T$.  It would appear that this happens with probability only $1/T$ for the standard construction, but this can be improved to any $\Omega(1)$ probability by the ``padding with the identity'' trick, e.g.\ by dovetailing the computation by a series of identity gates for $t=T,\ldots,4T$.  This padding increases the length of the clock (and hence the number of qubits needed to represent it) by a constant factor.  This is where the modified circuit Hamiltonian from \cref{theo:main} leads to an improvement: by tuning the couplings of the clock Hamiltonian, it is possible to create an arbitrarily high probability of finding the system in the final configuration at $t = T$, \emph{without} any increase in the size of the clock.  To show that this works we need to lower bound $\Delta_{\min}$ for this modified construction, resulting in the following theorem.
\begin{restatable}{theorem}{adiabatic}
	\label{theo:adiabatic}
	For any probability $p > 0$ there is a frustration-free modified circuit Hamiltonian as in \cref{eq:ham} with $\pi_T = p$ such that the minimum spectral gap of the linear interpolation \eqref{eq:adiabaticLinear} satisfies $\Delta_{\min} = \Omega(T^{-2})$.  	
\end{restatable}
\subsubsection{Tightness of the Geometrical Lemma}
A natural question is whether the lower bound given by the geometrical lemma is asymptotically tight---$\Omega(T^{-3})$, as stated e.g.\ in Kitaev's original paper \cite{Kitaev2002}---for (uniform weight) Feynman Hamiltonians.  
We show that this is not the case.
As a first step we apply an improved analysis to Kitaev's original construction with uniform weights,
\begin{align}
	\Hprop & = \sum_{t =0}^{T-1} \left (\ketbra{t}\otimes\1 + \ketbra{t+1}\otimes\1 - \ketbra{t+1}{t}\otimes \op U_t - \ketbra{t}{t+1}\otimes \op U^\dagger_t \right),\label{eq:standardHam} \\
	\ket\psi          & = \frac{1}{\sqrt{T+1}}\sum_{t = 0}^T \ket t \otimes \left(\op U_t \ldots \op U_1 \right) \ket\xi \label{eq:standardHistoryState}
\end{align}
and we find that $E(\HFK) = \Omega(T^{-2})$ for $\Hprop$ as in \cref{eq:standardHam}.
We capture this in the following theorem.
\begin{restatable}{theorem}{kitaev}\label{theo:kitaev}
	The UNSAT penalty in Kitaev's local Hamiltonian construction is $\Omega(T^{-2})$.
\end{restatable}
Note that this result has been independently proven using Jordan's lemma type techniques by \citeauthor{Caha2017} \cite{Caha2017} and us \cite{BauschThesis}.
In this work, we include a direct proof not based on Jordan's lemma, which can be found in \cref{sec:direct-proof}.

\subsubsection{Limitations on Further Improvements}\label{sec:limitationsIntro}
Finally, by applying the lower bound in \cref{lem:geo} to modified Feynman Hamiltonians of the form in \cref{eq:ham}, we show that the scaling of the UNSAT penalty obtained in \cref{theo:main} is the optimal scaling that can be achieved.
\begin{restatable}{theorem}{tridiag}
	\label{theo:tridiag}
	Let $\ket\psi$ be the ground state of a Hamiltonian $\op H$ with eigenvalues $E_0 \leq E_1 \leq \ldots \leq E_T$.
	If $\op H$ is tridiagonal in the basis $\{\ket0,\ldots,\ket T\}$,
	\begin{equation*}
		\op H := \sum_{t = 0}^T a_t \ketbra t + \sum_{t =0}^{T-1} \left ( b_t\ketbra{t+1}{t} + b^*_t\ketbra{t}{t+1} \right),
	\end{equation*}
	with $|a_t|, |b_t| \leq 1$ for $t = 0,\ldots,T$, then the product $\Delta_{\op H} \cdot \min \{ |\psi|^2_0,  |\psi_T|^2\} $ of the spectral gap $\Delta_{\op H} = E_1 - E_0$ and the minimum endpoint overlap is $\BigO(T^{-2})$.
\end{restatable}

\begin{figure}
  \centering
  \begin{tikzpicture}[node distance = 14mm, text height = 1.5ex, text depth = .25ex]
    \node (5) at (0,0) {5};
    \node (4) at (2,0) {4};
    \node (1) at (4,1) {1};
    \node (3) at (4,-1) {3};
    \node (2) at (6,0) {2};
    \draw[->] (1) to[in=135,out=0] node[above] {$\op U_1$} (2);
    \draw[->] (2) to[in=0,out=-135] node[below,xshift=8] {$\op U_2$} (3);
    \draw[->] (3) to[in=-45,out=180] node[below,xshift=-3] {$\op U_3$} (4);
    \draw[->] (4) to[in=180,out=45] node[above,xshift=-3] {$\1_d$} (1);
    \draw[->] (4) -- node[above] {$\1_d$} (5);
  \end{tikzpicture}
  \caption{A graph $G$ with five vertices $\{1,2,3,4,5\}$; if the assign unitaries to the edges---in this case three non-trivial ones $\op U_1, \op U_2,\op U_3\in\mathrm{SU}(\field C^d)$, and $\1_d$ everywhere else---we call the graph a unitary labeled graph (ULG). Every edge $(a,b)$ with unitary $\op U_{ab}$ is translated into a transition term $\op h_{ab}:=\sum_i(\ket a\ket i-\ket b\otimes \op U_{ab}\ket i)(\mathrm{h.c.})$, where the $a$ and $b$ take the role of the time register. If the product of unitaries along any loop---in this case only $\op U_3\op U_2\op U_1$---equals the identity $\1_d$, the circuit Hamiltonian built from these transition terms $\op H_G:=\sum_{(a,b)\in G}\op h_{ab}$ is unitarily equivalent to $\Delta\otimes\1_d$, where $\Delta$ is the graph Laplacian of the underlying graph $G$, as in \cref{lem:clockham}.}
  \label{fig:ulg-intro}
\end{figure}
Can this limitation for tridiagonal clock Hamiltonians be overcome using a clock that contains branches and loops?  In \cite{Bausch2016} it was shown that the circuit-to-Hamiltonian construction can in principle be considered on arbitrary graphs, with the spectral properties of these circuit Hamiltonians reducing to the analysis of the underlying graph.  These unitary labeled graphs define circuits which evolve according to unitary gates along their edges, see \cref{fig:ulg-intro}.   In this work we also consider generalized ULGs with weights along the vertices and edges, i.e.\ general Hermitian matrices, extending the graph Laplacian case that was analyzed in \cite{Bausch2016}. Given a Hamiltonian with the decomposition $\op H = \sum_{i,j \in \mathcal{B}} H_{ij} \ketbra{i}{j}$ in a particular basis, the generalized ULGs based on $\op H$ will have the form
$$
\Hprop = \sum_{i,j\in \mathcal{B}} H_{ij} \ketbra{i}{j} \otimes \op U_{ij}.
$$

 A ULG is called simple if the product of unitaries around any loop is the identity, and otherwise it is said to be frustrated.  For a simple ULG $\op H_{\textrm{prop}}$ based on a Hamiltonian $\op H$, a generalization of the standard argument shows that there will be a unitary transformation $\op W$ such that $\op W^\dagger \Hprop \op W = \op H \otimes \1$. 

In order to analyze the UNSAT penalty for a large class of generalized ULGs, we first make an observation that relates the low energy spectrum to the UNSAT penalty.  
\begin{remark}
Suppose $\op H_{\textrm{prop}} = \sum_{i,j \in \mathcal{B}} H_{ij} \ketbra{i}{j} \otimes \op U_{ij}$ describes a generalized ULG, and $\op H_{\textrm{in}}:= \sum_{i\in V_{\textrm{in}}} \ketbra i\otimes \Pi_{\textrm{in},i}$ and $\op H_{\textrm{in}}:= \sum_{i\in V_{\textrm{out}}} \ketbra i\otimes \Pi_{\textrm{out},i}$ are input and output projectors.  If $\op H_{\textrm{prop}}$ has $k\geq \min \left \{|V_{\textrm{in}}|,|V_{\textrm{out}}| \right\}  + 1$ excited eigenvalues of energy at most $R$, then the UNSAT penalty satisfies $E\left(\op H_{\textrm{in}} + \op H_{\textrm{prop}} + \op H_{\textrm{prop}} \right)  \leq R$.
\end{remark}
This follows because the $k$ excited states can be used to form a superposition $\ket\phi$ with zero amplitude on $k-1$ sites.  This superposition will have $\bra \phi \op H_{\textrm{prop}}\ket\phi \leq R$, and if the amplitudes are zero on all the input (output) penalties, then $ \bra{\phi}\left(\op H_{\textrm{in}} + \op H_{\textrm{prop}} + \op H_{\textrm{prop}} \right) \ket{\phi} = \bra\phi \op H_{\textrm{prop}}\ket\phi \leq R$.  In the case of an output penalty at a single vertex in the graph (as in the standard tridiagonal case), this shows that the spectral gap will always upper bound the UNSAT penalty. 

To apply this spectral upper bound on the UNSAT penalty, we make the hypothesis that any generalized ULG which describes computations of size $T$ must contain a path of length at least length $T$.  In other words, we start from an assumption that we are given circuits of size $T$ that are ``optimally compiled'', i.e.\ ones that and cannot be computed using local gates by a shorter path.
For a given Hamiltonian $\op H$, we further define the notion of diameter with respect to particular basis $\mathcal{B}$, which is defined to be the smallest power of $\op H$ needed to connect any two basis elements:
$$
\diam_{\mathcal{B}} \op H :=\max_{u,v \in \mathcal{B}} \min \left \{k \in \field N : \bra u \op H^k \ket v \neq 0\right \}.
$$
The ground state $\ket\psi$ of $\op H$ defines a probability distribution $\pi(x) := |\braket{x}{\psi}|^2$ for $x \in \mathcal{B}$.  Let $\pi_\mathrm{min} := \min_{x \in \mathcal{B}} \pi(x)$ and define the spectral gap to be $\Delta_{\op H} = \lambda_1 - \lambda_{\min}$.
This allows us to derive a bound for Hermitian matrices which relates the ground state probability distribution $\pi$ and the matrix diameter to the spectral gap, generalizing a known bound for graph Laplacians~\cite{KelnerMITLecturenotes2009}. 
\begin{restatable}{theorem}{thdiam}\label{th:diam}
 Let $\op H$, $\mathcal{B}$, $\pi$ and $\Delta_{\op H}$ be as defined above. Then
\[ 
	\Delta_{\op H} \leq \frac{\|\op H\|}{2} \left( \frac{\log(2/\pi_{\min})} {\diam_{\mathcal{B}}\op H}\right)^2.
\]
\end{restatable}
Since the diameter is at least $T$ under the optimal compilation assumption, obtaining a spectral gap that scales like $T^{-2 + \epsilon}$ for some $\epsilon > 0$ requires a generalized ULG Hamiltonian with either a growing operator norm, or the presence of exponentially small ground state amplitudes.  The latter case is typically ruled out by the fact that exponentially small amplitudes comprise points at which it is easy to deform the computation and obtain a small UNSAT penalty.

\subsection{Related Work}
The notion of modifying the weights in Feynman's circuit Hamiltonian goes back to at least 1985, when Peres~\cite{Peres1985} considered modified weights for improving the probability of measuring the output of the computation in the ballistic Hamiltonian model.  In the context of universal adiabatic computation Ganti and Somma~\cite{ganti2013gap} derived a limitation on improving the spectral gap of circuit Hamiltonians by applying the lower bound on the Gover search problem.  In \cite{Usher2017}, the authors go beyond unitary circuit include certain projective measurements $\{\Pi,\1-\Pi\}$ (ones whose outcome probabilities are independent of all valid input state at that computational step).  

More recently, \citeauthor{Caha2017} have presented a compilation of several results on modified circuit Hamiltonians~\cite{Caha2017}.  Most of these results are distinct from results in this work, with the exception of \cref{theo:kitaev} which was discovered by us independently
around the same time \cite{BauschThesis}.
Another aspect of \cite{Caha2017} and the present work that can be compared are that both works make an improvement to idling the computation efficiently.  \citeauthor{Caha2017} show that the overlap with the final state of the quantum circuit can be increased arbitrarily close to 1 in the Hamiltonian computation model by using logarithmic overhead in space.  We show a similar result for universal adiabatic computation, with no space overhead but at the expense of requiring higher precision in the Hamiltonian couplings.  

\pagebreak[4]

\section{Non-Uniform Circuit Histories for Improving the UNSAT Penalty}
\subsection{Analysis of Modified Feynman Hamiltonians}\label{sec:analysis}
\begin{proofsomething}{lem:clockham} The lemma, in summary, claims that $\Hprop$ is unitarily equivalent to copies of $\Hclock$; we re-state it here for convenience.
\clockham*
\begin{proof}
Since $\op W$ is a linear operator the calculations we need to check for \cref{lem:clockham} are the same as in the standard unweighted case~\cite[ch.~14.4]{Kitaev2002}.
As a reminder,
\[
	\op W^\dagger \op W  = \sum_{t,t'=0}^T \Big(\proj t \otimes (\op U_1^\dagger \cdots \op U_t^\dagger)\Big)\Big(\proj{t'} \otimes (\op U_1\cdots \op U_{t'}) \Big) = \sum_{t=0}^T \proj t \otimes \1  = \1,
\]
\[
	\op W^\dagger(\ketbra{t+1}{t}\otimes  \op U_{t+1})\op W  = \ketbra{t+1}{t} \otimes (\op U_1^\dagger \cdots \op U_{t+1}^\dagger) \op U_{t+1} (\op U_t\cdots \op U_1)  = \ketbra{t+1}{t} \otimes \1,
\]
and so the claim of \cref{lem:clockham} follows by linearity.
\end{proof}
\end{proofsomething}

\begin{proofsomething}{lem:geo}
Kitaev's geometrical lemma provides the starting point for the lower bound on the UNSAT penalty given in  \cref{eq:geo}.
\begin{lemma}[Kitaev's geometrical lemma]\label{lem:kitaev}
	Let $\op A,\op B\ge0$ be positive semi-definite operators, both with a zero eigenspace, and such that $\ker \op A\cap\ker \op B=\{0\}$.
	Denote with $\lnz(\op A),\lnz(\op B)$ the minimal non-zero eigenvalue of $\op A$ and $\op B$, respectively.
	Then
	\begin{equation}
		\op A+\op B \ge \min\{\lnz(\op A),\lnz(\op B)\}\times 2\sin^2\frac\theta2,\label{eq:kit}
	\end{equation}
	where $\theta$ is the angle between the kernels of $\op A$ and $\op B$.
\end{lemma}

This allows us to proof the lower bound, re-stated here for convenience:
\geo*
\begin{proof}
We use the notation from \cref{lem:kitaev}.
For us, $\op A=\Hin + \Hout$, and $\op B=\Hprop$,  and in this section we use the freedom to shift the energy in \cref{eq:ham} to set $E(\Hprop) = 0$ (since the system can be frustrated, this means the local terms may no longer be positive semi-definite, but this will not present a problem in applying the geometrical lemma above because $\Hprop$ itself is positive semi-definite).
Denote the projector onto the kernel of the penalty terms $\op A$ with  $\Pi_\mathrm{pen}:=\ketbra 0\otimes\Pin^\perp+\ketbra T\otimes\Pout^\perp+\sum_{t=2}^{T-1}\ketbra t\otimes\1$.
Denote with $\op U=\op U_T\cdots \op U_1$ the entire encoded quantum circuit.
We first want to bound the angle $\theta$ between the kernels of the propagation and penalty Hamiltonians\cite{Kitaev2002,CubittLectureNotes}.

\begin{align*}
	\cos^2\theta                                                 & =\max_{\substack{\ket{\xi}\in\ker \op A                                                                                                                                                             \\
	\ket{\eta}\in\ker \op B}}|\braket{\xi}{\eta}|^2              &  \\
	                                                             & =\max_{\substack{\ket{\xi}                                                                                                                                                                          \\
	\ket{\eta}\in\ker \op B}}|\bra\eta\Pi_\text{pen}\ket{\xi}|^2 &  \\
	                                                             & \overset*=\max_{\ket{\eta}\in\ker \op B}\bra\eta\Pi_\text{pen}\ket{\eta}                                                                                                                            \\
	                                                             & =\max_{\ket\eta\in\ker \op B}\bra\eta \op W^\dagger( \op W\Pi_\text{pen}
	\op U^\dagger) \op W\ket{\eta}                                                                                                    \\
	                                                             & =\max_{\ket{\eta'}\in\ker \op W\op B \op W^\dagger}\bra{\eta'}\left(\ketbra 0\otimes\Pin^\perp+\ketbra T\otimes \op U\Pout^\perp \op U^\dagger+\sum_{t=2}^{T-1}\ketbra t\otimes\1\right)\ket{\eta'} \\
	                                                             & =\max_{\ket{\phi}}\sum_{s=1}^T\psi_s^*\psi_s\bra s\bra\phi\left(\ketbra 0\otimes\Pin^\perp+\ketbra T\otimes \op U\Pout^\perp \op U^\dagger+\sum_{t=2}^{T-1}\ketbra t\otimes\1\right)\ket t\ket\phi  \\
	                                                             & =\max_{\ket{\phi}}\bra\phi(\abs{\psi_0^2}\Pin^\perp+\abs{\psi_T^2} \op U\Pout^\perp \op U^\dagger)\ket\phi + 1-\abs{\psi_0^2}-\abs{\psi_T^2},
\end{align*}
where we have saturated Cauchy-Schwartz in the third line $(*)$.
To bound the first inner product, we observe that if $\psi_0^2\ge\psi_T^2$, picking $\ket{\phi}\in\ker\Pin$ gives the bound
\begin{equation*}
	\max_{\ket{\phi}}\bra\phi(\pi_0\Pin^\perp+\pi_T \op U\Pout^\perp \op U^\dagger)\ket\phi \le \pi_0 + \pi_T\cos\vartheta, \end{equation*}
where $\vartheta$ is the angle between $\supp \Pin$ and $\supp \op U\Pout \op U^\dagger$.
This angle can be lower-bounded by the acceptance probability of the circuit:
\begin{equation*}
	\cos^2\vartheta = \hspace{-1.8cm}\max_{\subalign{\hspace{2cm}\ket\eta&\in\supp \Pin\\\ket\xi&\in\supp \op U\Pout \op U^\dagger}}\hspace{-.5cm} |\braket{\eta}{\xi}|^2 \le \hspace{-1.2cm}\max_{\subalign{\hspace{1.4cm}\ket\eta&\in\supp \Pin\\\ket\xi&\in\supp \Pout}} |\bra{\eta}\op U\ket{\xi}|^2 \le \epsilon.
\end{equation*}
Similarly, if $\psi_0^2<\psi_T^2$, one can show an upper bound of $\pi_0\cos\vartheta + \pi_T$.
We thus obtain an overall upper bound
\begin{align*}
	\cos^2\theta & \le \max\{\pi_0,\pi_T\} + \min\{\pi_0,\pi_T\}\sqrt\epsilon + 1 - \pi_0 - \pi_T \\ & \le 1  - \min\{\pi_0,\pi_T\}(1-\sqrt\epsilon).
\end{align*}
In Kitaev's lemma, we thus obtain a lower bound
\begin{equation*}
	2\sin^2\frac\theta2 \ge 2\times\frac{1-\cos^2\theta}{8\cos^2\theta} \ge \frac14\min\{\pi_0,\pi_T\}(1-\sqrt\epsilon),
\end{equation*}
and the claim follows.
\end{proof}
\end{proofsomething}

\subsection{Metropolis Hamiltonians with Target Ground State Distributions}\label{sec:metro}
\begin{proofsomething}{lem:mcham}
We assume the reader has some familiarity with Markov chain transition matrices; an introduction can e.g.\ be found in \cite{levin2009markov}.
We re-state the lemma as follows:
\mcham*
\begin{proof}
Given a probability distribution $\pi$ with support everywhere on its domain $\mathcal{S}= \{0,\ldots,T\}$, let $\op P$ be the Markov chain with Metropolis transition probabilities defined on $\mathcal{S}$,
\begin{equation}
	\op P_{t,t+1} = \frac{1}{4} \min\left\{1,\frac{\pi_{t+1}}{\pi_t}\right \} \quad , \quad \op P_{t,t-1}  = \frac{1}{4}\min\left\{1,\frac{\pi_{t-1}}{\pi_t}\right \} \quad \op P_{t,t} = 1 - \op P_{t,t+1} - \op P_{t,t-1}\label{eq:MetropolisTransitionProbabilities}
\end{equation}
for all $i \in \mathcal{S}$ (setting the expressions $\op P_{0,-1}$ and $\op P_{T,T+1}$ to zero) and $\op P_{t,t'} = 0$ for all $t,t'\in \mathcal{S}$ with $|t-t'| > 1$.
The choice of coefficient $1/4$ implies $\op P_{t,t} \geq 1/2$ for all $t$ and so $\op P$ is positive semi-definite.
The principal left eigenvector of this transition matrix $\op P := \sum_{t,t' \in \mathcal{S}} \op P_{t,t'}\ketbra{t}{t'}$ is $\bra\pi = \sum_{t} \pi_t \bra t$, and while $\op P$ is not a symmetric matrix, there is a well known similarity transformation that relates $\op P$ to a symmetric matrix, given by
\[
	\op A:= \sum_{t,t' \in \mathcal{S}}\pi_t^{1/2} \pi_{t'}^{-1/2} \op P_{t,t'}\ketbra{t}{t'}.
\]
The two matrices are related by the fact that if $\bra{v_0}, \cdots, \bra{v_T}$ are the left eigenvectors of $\op P$ with eigenvalues $\lambda_0 = 1 \geq \lambda_1 \geq , \ldots , \geq \lambda_T\geq 0$, then $\ket{w_i} := \sum_ {t\in \mathcal{S}} \braket{v_i}{t} \braket{t}{v_0}^{-1/2} \; \ket t$ satisfies $A \ket{w_i} = \lambda_i \ket{w_i}$.
Therefore $\op A$ has the same eigenvalues as $\op P$, and in particular it has largest eigenvalue 1 corresponding to the eigenvector $\ket{w_0}$ with components satisfying $\braket{t}{w_0} = \braket{t}{v_0}^{1/2} = \sqrt{\pi_t}$.
Therefore $\op H = \1 - \op A$ is a nonnegative Hermitian matrix with ground state that has energy zero and components $\sqrt{\pi_t}$ in the time register basis, as claimed.
\end{proof}
\end{proofsomething}

\paragraph{Markov chain spectral gaps.}
Substantial effort has been devoted to characterizing spectral gaps of Markov chains.
A particularly fruitful characterization proceeds by defining a quantity called the conductance,
\[
	\Phi := \min_{S \subseteq \Omega} \frac{Q(S,S^c)}{\min\{\pi(S),\pi(S^c)\}} \quad , \quad Q(S,S^c) := \sum_{x \in S , y \in S^c } \pi(x) P(x,y)
\]
which determines the spectral gap within a quadratic factor,
\begin{equation}
	\frac{\Phi^2}{2} \leq \Delta_{\op P} \leq 2 \Phi.
	\label{eq:cheeger}
\end{equation}
The lower bound in \cref{eq:cheeger} is known as Cheeger's inequality, and it was initially discovered in the analysis of manifolds~\cite{Cheeger1970} before being adapted to the setting of Markov chains~\cite{sinclair1989approximate}.
In the next section we will use this method to lower bound the spectral gap of the Metropolis Hamiltonian corresponding to a particular non-uniform stationary distribution.

\subsection{Explicit Construction of an $\Omega(T^{-2})$ UNSAT Penalty Circuit Hamiltonian}\label{sec:theo1proof}
\begin{proofsomething}{theo:main} Using techniques developed in \cref{lem:mcham}, we can proof the first of our main theorems, which we re-state.
\main*
\begin{proof}
Let $\Hprop$ be the Metropolis Hamiltonian from \cref{sec:metro} corresponding to the probability distribution
\[
	\pi_0 = \pi_T = \frac{1}{4} \quad \textrm{and} \quad \pi_t = \frac{1}{2(T-1)} \quad \textrm{for} \quad t = 1,\ldots,T-1.
\]
A few low energy eigenstates of $\Hprop$ are illustrated in \cref{fig:weights-markov}.
 Keeping with tradition~\cite{aharonov2002quantum, aharonov2008adiabatic, Aharonov2009}, we exhibit $\Hprop$ as a $T+1$ by $T+1$ matrix in the clock register basis,
\setstackgap{L}{32pt}
\setstacktabbedgap{-4pt}
\fixTABwidth{T}
\begin{equation}
	\label{eq:hammatrix} \Hprop = \dfrac12
	\parenMatrixstack{
		\dfrac{1}{T -1}             & - \dfrac{1}{\sqrt{2 T - 2}} & 0              &                & \cdots         &                           & 0                         \\
		- \dfrac{1}{\sqrt{2 T - 2}} & 1                           & - \dfrac{1}{2} & 0              &                &                           &                           \\
		0                           & - \dfrac{1}{2}              & 1              & \ddots         & \ddots         &                           & \vdots                    \\
		                            & 0                           & \ddots         & \ddots         & \ddots         & 0                         &  \\
		\vdots                      &                             & \ddots         & \ddots         & 1              & - \dfrac{1}{2}            & 0                         \\
		                            &                             &                & 0              & - \dfrac{1}{2} & 1                         & - \dfrac{1}{\sqrt{2T -2}} \\
		0                           &                             & \cdots         &                & 0              & - \dfrac{1}{\sqrt{2T -2}} & \dfrac{1}{T - 1}
	}.
\end{equation}

Since $\pi_0$ and $\pi_T$ are $\Omega(1)$ it only remains to be checked that the spectral gap $\Delta_{\Hprop}$ of the clock Hamiltonian is $\Omega(T^{-2})$.  The spectral gap between the lowest two eigenvalues of $\Hprop$ is equal to the spectral gap between 1 and the second largest eigenvalue of the Metropolis transition matrix $\op P$, so we can apply Cheeger's inequality \cref{eq:cheeger} to $\pi$ and $\op P$.

The goal is to show that every subset $S$ of $\{0,\ldots,T\}$ has large conductance, so we divide the proof into cases corresponding to the different possibilities for the subset $S$.
First if $S = \{0\}$ then since $\op P_{0,1} = (8T - 8)^{-1}$ so $\Phi(S)$ is $\Omega(T^{-1})$, with similar statements holding for $S = \{T\}$ and $S = \{0,T\}$.
Now if $S \subseteq \{1,\ldots,T-1\}$ is non-empty there must be at least one $t \in \{1,\ldots,T-1\} $ such that there is a $t \in S^c$ with $\op P_{t,t'} \geq 1/4$, and since $\pi_t = (2T -2)^{-1}$, this shows that $\Phi(S)$ is $\Omega(T^{-1})$ in this case as well.

Therefore $\Delta_{\op P}$ is $\Omega(T^{-2})$ by \cref{eq:cheeger}, and so $\Delta_{\Hprop} = \Omega(T^{-2})$ as well.  Using \cref{lem:geo} this concludes the proof.
\end{proof}
\end{proofsomething}

\section{Universal Adiabatic Computation}\label{sec:adiabatic}
The circuit-to-Hamiltonian construction used in the proof of QMA-completeness also plays a crucial role in the proof that adiabatic quantum computation (AQC) can simulate the quantum circuit model with polynomial overhead.
In this section we will review the relevant aspects of the standard universal AQC construction in order to explain how our results on modified Feynman Hamiltonians allow for a useful practical improvement in universal AQC.
Furthermore, we show that \cref{theo:tridiag} yields a lower bound on the time needed for universal AQC using modified Feynman Hamiltonians to simulate the circuit model.
\begin{figure}
	\includegraphics[width=\textwidth]{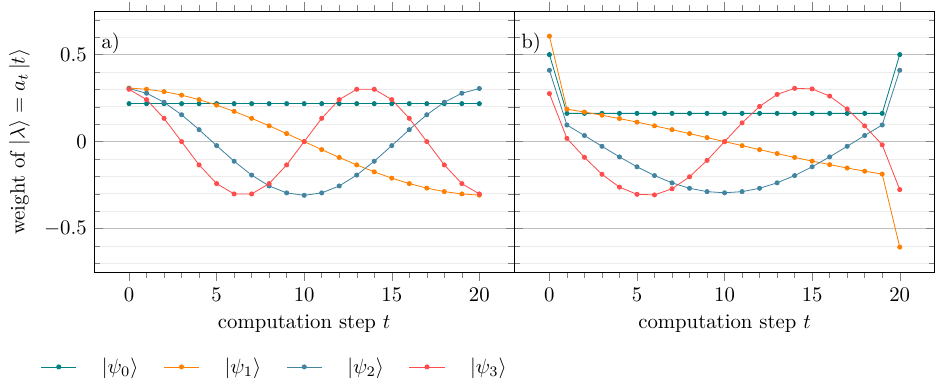} \caption{ Low energy eigenstates of \textbf{(a)} the standard circuit-to-Hamiltonian construction and \textbf{(b)} the Metropolis Hamiltonian corresponding to the distribution with $\pi_0,\pi_T = 1/4$ that is used in this section.
	}
	\label{fig:weights-markov}
\end{figure}
The standard version of universal AQC is based on $\Hprop$ as given in \cref{eq:standardHam}, together with
\newcommand{\Hinit}{\op H_\textnormal{init}}
\[
	\Hinit := \ketbra 0 \otimes \1 =
	\begin{pmatrix}
		0                       & 0      & \multicolumn{2}{c}{\cdots} & 0                       \\
		0                       & 1      & \ddots            &        & \multirow{2}{*}{\vdots} \\
		\multirow{2}{*}{\vdots} & \ddots & \ddots            & \ddots &  \\
		                        &        & \ddots            & 1      & 0                       \\
		0                       & \multicolumn{2}{c}{\cdots} & 0      & 1
	\end{pmatrix}
	\otimes \1.
\]
The goal in universal AQC is to prepare the ground state of $\Hprop$ by continuously varying a Hamiltonian $\op H(s)$ that depends on an adiabatic parameter $0 \leq s \leq 1$.
First, when $s = 0$, the system is initialized in the ground state of $\op H(0) := \Hinit$, which is a fiducial state that can be prepared easily.
Then the parameter $s$ is slowly increased until it reaches $s = 1$, at which point the system Hamiltonian is $\op H(1) := \Hprop$.
Defining $\|\dot{\op H}\| = \max_s \|\dd \op H / \dd s\|$ and $\Delta = \min_s \Delta_{\op H}(s)$, an often used condition
for preparing the ground state of $\Hprop$ is for the total evolution time to satisfy $t_\textrm{adiabatic} = \Theta(\|\dot{\op H}\| \Delta^{-2})$ (see~\cite{albash2018adiabatic} for a review of rigorous adiabatic theorems).

This time scale was shown~\cite{aharonov2008adiabatic} to be bounded by $\poly(n)$ for the linear interpolation schedule, i.e.
\begin{equation}\label{eq:standardAdiabatic}
	\op H(s) = (1-s) \Hinit + s \Hprop,
\end{equation}
where specifically $\Delta = \Omega(T^{-2})$ and $\|\dot{\op H}\| = O(1)$ so that $t_\textrm{adiabatic} = \Theta(T^4)$.
The lower bound on the spectral gap was proven using a quantum-to-classical mapping, a monotonicity property of the ground state wave function as a function of $s$, and Cheeger's inequality.

One may notice that measuring the time register of the uniform history state in \cref{eq:standardHistoryState} results in collapsing the computational register to its final time step $t = T$ with probability $O(T^{-1})$.
To avoid repeating the adiabatic evolution many times, it is desirable to boost this probability to $\Omega(1)$.
The standard trick for doing this is to pad the end of the computation with identity gates, meaning after performing the desired computational gates $\op U_1,\ldots,\op U_T$ one adds $r$ additional identity gates $\op U_r = \op U_{r+1} = \ldots = \op U_{r+T} = \1$, so that measuring $t\in [T,r+T]$ suffices to project the computational register to be in the end of the computation.

	\begin{figure}
		\centering
		\includegraphics[width=10cm]{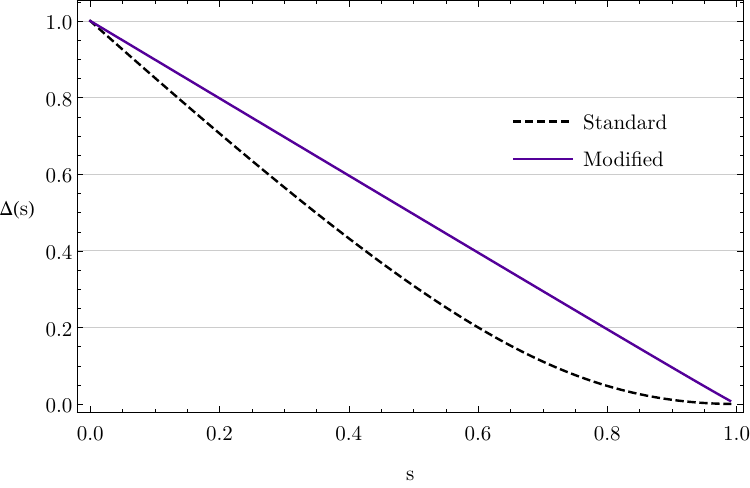}
		\caption{
			Comparison of the spectral gap as a function of the adiabatic parameter $s$ for the standard version of universal AQC with a $1/T$ probability of measuring the history state to be in the final time step of the computation and for the modified construction corresponding to a history state with probability $1/4$ of measuring the final time step of the computation.
			Here  $T = 100$ for both constructions, and they both follow the same adiabatic schedule with local terms of the same norm.
		}
		\label{fig:adiabatic}
	\end{figure}

We want to point out that modified Feynman Hamiltonians, together with the Metropolis Hamiltonian construction of \cref{sec:metro}, open up a new set of trade-offs in universal adiabatic computation that may be relevant for practical implementations.
Specifically, let $\op H^*_\textrm{prop}$ be the Hamiltonian used in \cref{sec:theo1proof}.  The $\Omega(1)$ probability on the endpoints can be used to increase the probability that measuring the time register will collapse the computational register of the system into the final time step of the computation.
This provides an alternative to the ``padding the end of the computation with identity gates'' technique that is normally used to raise the probability of sampling from the final time step of the computation.
Padding the length of the computation with identity gates is relatively expensive in practical terms when the time register is encoded using local interactions (such as the domain wall clock) because the clock must be represented in unary, meaning the number of clock qubits scales linearly with the total length of the (padded) computation.

Achieving an overlap of $\delta \approx 1 $ with the final step of the computation by padding the system with identity gates requires a total of $\BigO(T/(1-\delta))$ clock qubits, however one can instead prepare the weighted history state with $\pi_T = \delta$ using only $T$ clock qubits.
The price that one has to pay for this improvement is in an increase in the precision of the couplings needed to implement \cref{eq:ham} now must scale like $\BigO(T^{-1})$, as seen in \cref{eq:hammatrix}.
This is a reasonable trade off, however, since the total number of qubits is generally the limiting factor in most experiments.

\begin{proofsomething}{theo:adiabatic}
Having seen that it can be advantageous to prepare the non-uniform history state ground state of $\op H^*_\textrm{prop}$, it remains to be shown that this can be done efficiently, which is what the second of our theorems claims:
\adiabatic*
\begin{proof}
In \cref{sec:theo1proof} we proved that $\Delta_{\op H^*_\textrm{prop}} = \Theta(T^{-2})$, just as the minimum gap for the standard construction is $\Delta_{\op H(s = 1)} = \Theta(T^{-2})$.
A numerical comparison of the spectral gaps of $\op H(s)$ in \cref{eq:standardAdiabatic} and of $(1-s)\Hinit + s \Hprop^*(s)$ can be found in \cref{fig:adiabatic}, which reveals that both of these linear interpolations encounter their asymptotic minimum gap at $s = 1$.

However, it is of course desirable to bound the adiabatic run time for $\op H^*(s)$ without any recourse to numerics.
In this case the monotonicity of the ground state wave function used in the proof for the standard construction does not hold, and so we offer a different proof based on a modified linear interpolation schedule,
\begin{equation}\label{eq:modifiedAdiabatic}
	\op H^*(s) = \op H_0 + s A \op H^*_\textnormal{prop}.
\end{equation}
Choosing $A = T^4$ and an $\Omega(1)$ spectral gap for $0 \leq s \leq 1$, the standard estimate of the adiabatic run time $\Omega(\|\dot{\op H}\| \Delta_{\op H}^{-2})$ is then $\Omega(T^4)$---just as it will be for the usual version of universal AQC.
At $s = 1$, the system will be in the ground state of the Hamiltonian $\tilde{\op H}_\textnormal{final} = \op H^*_\textnormal{prop} + A^{-1} \op H_0$.  To bound the distance from this state to the target ground state of $\op H^*_\textnormal{prop}$ we apply the following general result from Weyl perturbation theory.
\begin{proposition}
	Let $\op H = \op H_0 + \lambda \op V$ be a perturbed system, with $\op V \succeq 0$ and the unperturbed Hamiltonian $\op H_0$ having a spectral gap $\Delta_0$ between its ground space and the rest of the spectrum.  For any ground state $|\psi\rangle$ of $\op H$ there is a ground state $|\psi_0\rangle$ of $\op H_0$ with fidelity
$$
|\langle \psi | \psi_0\rangle | ^2 \geq  1 - \frac{2 \lambda \|\op V\| } {\Delta_0}.
$$
\end{proposition}
In order to show the proposition, we start with Weyl's inequality.
The ground state and first excited energies of the perturbed system satisfy
$$
E_0 \geq E_0^0 - \lambda \|\op V\| ,   \quad , \quad E_1 \geq E_0^1 - \lambda \|\op V\|.
$$
Meanwhile, any original ground state $|\psi^0_0\rangle$ has $\langle \psi^0_0 | \op H | \psi^0_0\rangle \leq E^0_0 + \lambda \|\op V\|$.  We can express $|\psi^0_0\rangle$ as a combination of a ground state $|\psi\rangle$ of $\op H$ and an orthogonal state $|\psi_\perp\rangle$, 
$$
|\psi_0^0\rangle = \sqrt{1- \epsilon}|\psi\rangle + \sqrt{\epsilon} |\psi_{\perp}\rangle .
$$
The energy of this state satisfies $\langle \psi_0^0| \op H |\psi_0^0\rangle \leq  E^0_0 + \lambda \|\op V\|$, and so
\begin{align*}
 E^0_0 + \lambda \|\op V\|& \geq \langle \psi_0^0| \op H |\psi_0^0\rangle\\
&\geq (1 -\epsilon) E_0 + \epsilon E_1\\
& \geq (1 -\epsilon) (E_0^0 - \lambda \|\op V\| ) + \epsilon (E_0^1 - \lambda \|\op V\|)\\
& =  (E_0^0 - \lambda \|\op V\| ) + \epsilon (E_0^0 - E_0^1 ).
\end{align*}
Rearranging yields
$$
\epsilon \leq \frac{2 \lambda \|\op V\|} {E_0^0 - E_1^1}.
$$

Since $A = T^4$ we have $\left|\bra{\tilde{\psi}}\ket{\psi^*_0} \right|^2 = 1 - \BigO(T^{-2})$, as claimed.

To see that $\op H^*(s)$ in \cref{eq:modifiedAdiabatic} has $\Delta_{\op H^*(s)} = \Omega(1)$ for all $s$, first note that $\Delta_{\op H^*(0)} = 1$.
We want to show that the first excited energy is non-decreasing as a function of $s$.
In order to achieve this, we let $L$ be a large integer and discretize the adiabatic path into steps $s_0 = 0,s_1, \ldots, s_L = 1$, with a uniform step size $s_{i+1} = s_i + L^{-1}$.
Since $\op H^*(s_{i+1}) = \op H^*(s_i) + \left(\op H^*(s_{i+1}) - \op H^*(s_i)\right)$ and since $\op H^*(s_{i+1}) - \op H^*(s_i) = \Hprop^*$ is positive semi-definite, we can apply Weyl's inequality to see that the first excited state energy is indeed non-decreasing with $s$; because this holds for any $L$ this statement is independent of the discretisation of the path: we conclude that $E^*_1(s) \geq 1$.
Next we apply the variational method with a trial state $\ket\phi$ equal to the ground state of $\Hprop^*$, to see that $\bra{ \phi} \op H^*(s) \ket{ \phi} = \bra{ \phi } \op H_0 \ket{ \phi} = 3/4$, and so $\Delta_{\op H^*(s)} \geq 1/4$ for all $s$.
\end{proof}
\end{proofsomething}

\section{Tightness of the Geometrical Lemma for the UNSAT Penalty}\label{sec:tight-bounds}
\subsection{A Tight Bound for the Clock Hamiltonian}
In this section we prove that the UNSAT penalty of Kitaev's original construction is actually $\Omega(T^{-2})$, rather than the $\Omega(T^{-3})$ which can be shown using the geometrical lemma.  As noted in the introduction, this result was proven independently by \citeauthor{Caha2017}~\cite{Caha2017}, and appears to have been known for a while \cite{NagajTalk}.
We present a variant not reliant on Jordan's lemma in \cref{sec:direct-proof}.

The Hamiltonian in Kitaev's original proof is $\Hkit := \Hprop + \ketbra 0\otimes\Pin + \ketbra T\otimes\Pout$, where $\Hprop$ is given by \cref{eq:hamKitaev} with $a_0=a_T=1$ and $a_t=2$ else, and $b_t=-1$ for all $t$.
We denote the encoded circuit with $\op U=\op U_T\cdots\op U_1$.
The corresponding clock- and computation registers then live on the Hilbert space $\mathcal H:=(\field C^{T+1})\otimes(\field C^2)^{\otimes d}$ for some local dimension $d$.

As a first step, and to establish some background machinery, we focus on only the \emph{clock} part of the Hamiltonian, i.e.\ we disregard the encoded computation completely; the Hamiltonian we analyse has the form $\op H = \proj 0 + \Hclock + \proj T$.
Since $\op H$ is stoquastic, we can lower bound the ground state energy of $\op H$ by combining a suitable ansatz for the ground state with the following lemma.

\begin{lemma}[\citeauthor{Farhi2011}, \cite{Farhi2011}]
	\label{lem:lb}
	Let $\op H$ be a Hermitian operator which is stoquastic in the $\ket{t}$ basis (meaning $\bra{ t} \op{H}\ket{t'} \leq 0$ for all $t \neq t'$).
	Let $ E$ be its lowest eigenvalue.
	Then
	\begin{equation}
		E \ge \min_t \frac{\bra t \op H\ket\phi}{\braket{t}{\phi}}\label{eq:gap-lb}
	\end{equation}
	for any state $\ket\phi$ such that $\braket{t}{\phi} > 0$ for all $t$.
\end{lemma}
Noting that the state $\ket\phi$ in Lemma \cref{lem:lb} does not need to be normalized, define
\[
	\ket\phi := \sum_{t = 0}^T \sin \left(\frac{\pi (t+1)}{T+2} \right ) \ket t .
\]

The bound \cref{eq:gap-lb} can be evaluated using three cases.
The first case is $t = 0$, which yields
\begin{align*}
	\frac{\bra 0 \op H\ket\phi}{\braket{ 0 }{ \phi}} & = \left [\sin\left(\frac{\pi}{T+2}\right) \right]^{-1} \left[2\sin\left (\frac{\pi}{T+2} \right ) - \sin\left (\frac{2 \pi}{T+2}\right ) \right] \\ &= 2 \left (1 - \cos \left (\frac{\pi}{T+2} \right) \right ) \\ &= \frac{\pi^2}{T^2} -  \BigO(T^{-3}) .
\end{align*}

The next case is $1 \leq t \leq T - 1$,
\begin{align*}
	\frac{\bra t \op H\ket\phi}{\braket{ t }{ \phi}} & = \left [\sin\left(\frac{\pi(t+1)}{T+2}\right) \right]^{-1} \left [2  \sin\left(\frac{\pi(t+1)}{T+2}\right) -  \sin\left(\frac{\pi t}{T+2}\right) -  \sin\left(\frac{\pi(t+2)}{T+2}\right) \right ] \\ &=  2 \left (1 - \cos \left (\frac{\pi}{T+2} \right) \right ) \\ &= \frac{\pi^2}{T^2} -  \BigO(T^{-3}) .
\end{align*}

The final case is $t = T$,
\begin{align*}
	\frac{\bra T \op H\ket\phi}{\braket{ T }{ \phi}} & = \left [\sin\left(\frac{\pi(T+1)}{T+2}\right) \right]^{-1}\left[2\sin\left (\frac{\pi(T+1)}{T+2} \right ) - \sin\left (\frac{\pi T}{T+2}\right ) \right] \\ &=2-\sin \left(\frac{\pi  T}{T+2}\right) \csc \left(\frac{\pi  (T+1)}{T+2}\right) \\ &= \frac{\pi^2}{T^2} -  \BigO(T^{-3})
\end{align*}
We have thus shown that $\bra{ t } \op H \ket{\phi} / \braket{ t }{ \phi }$ is $\Omega(T^{-2})$ for all $t$, and so applying  \cref{lem:lb} we can conclude that $E(\op H)$ is $\Omega(T^{-2})$.

\subsection{A Tight Bound for the Full Circuit Hamiltonian}
\subsubsection{Block-Decomposition of the Full Circuit Hamiltonian}\label{sec:jordan-proof}
For the purpose of this discussion, we will assume that $\rank\Pin,\rank\Pout \ge d/2$; this is a natural assumption e.g.\ if $\Pin=\ketbra 1^{(1)}\otimes\1^{(2)}\otimes\cdots\otimes\1^{(d)}$ just penalizes the first qubit in a state $\ket1$.
Note that more penalized qubits generally \emph{increase} the rank of $\Pin$ if we demand the penalties to be local operators (i.e.\ not something like $\ketbra1 \otimes \ketbra 0 \otimes \cdots$, but $\ketbra1\otimes\1 + \1\otimes\ketbra0\otimes\1 + \ldots$), so this assumption is justified.

We know that there exists a global unitary $\op W$ on $\mathcal H$ such that $\spec(\Hprop)=\spec(\op W^\dagger\Hprop\op W)=\spec(\Delta)$ up to multiplicities, where $\Delta$ is the Laplacian of a path graph of $T+1$ vertices:
\begin{equation*}
	\Delta=
	\begin{pmatrix}
		1                       & -1 & 0                     & \multicolumn{2}{c}{\cdots} & 0                       \\
		-1                      & 2  & -1                    &        &                   & \multirow{2}{*}{\vdots} \\
		0                       & -1 & \ddots                & \ddots &                   &  \\
		\multirow{2}{*}{\vdots} &    & \ddots                & \ddots & -1                & 0                       \\
		                        &    &                       & -1     & 2                 & -1                      \\
		0                       & \multicolumn{2}{c}{\cdots} & 0      & -1                & 1
	\end{pmatrix}
\end{equation*} Then
\begin{equation}
	\label{eq:hamKitaev} \op W^\dagger\Hkit\op W = \underbrace{\Delta\otimes\1 + \ketbra0\otimes\Pin}_{:=\op A} + \underbrace{\ketbra T\otimes\op U^\dagger\Pout\op U}_{:=\op B}.
\end{equation}
We choose to work in a basis where $\op A$ takes the form
\begin{equation}
	\label{eq:A} \op A=\operatorname{diag}(\underbrace{\Delta',\ldots,\Delta'}_{\rank\Pin\ \text{times}},\Delta,\ldots,\Delta),
\end{equation}
and where $\Delta'=\Delta + \ketbra 0$.
Note that this is always possible: we simply re-order the computational register such that $\Pin$ penalizes the first $\rank\Pin$ states.
Also note that $\op B=\ketbra T\otimes\op U^\dagger\Pout\op U$ does not have a simple form in this basis (apart from having only a single entry \emph{within} each block in the time register basis, i.e.\ at the diagonal entry $\ketbra T$).

If we na\"{\i}vely try to diagonalize $\op B$, we will again mix up the nice block-diagonal form in \cref{eq:A}; our goal is thus to find a unitary $\op V=\1_d\otimes(\op V'\oplus\op V'')$ which respects the block-diagonal structure of $\op A$, i.e.\ in particular leaves the upper left block invariant (which is automatically the case if $\dim\op V'\le\rank\Pin$).

While Jordan's original paper addresses the case of orthogonal transformations between subspaces, we can phrase it more suitable to our needs:
\begin{lemma}[Jordan \cite{Jordan1875}, Th.~1]\label{lem:jordan}
	Let $\Pi_0$ and $\Pi_1$ be two Hermitian projectors in some Hilbert space \H.
	Then there exists a decomposition of $\H$ into one- and two-dimensional subspaces $\H=\bigoplus_i \H_i$, such that $\H_i$ is invariant under both $\Pi_0$ and $\Pi_1$, and such that $\rank\Pi_j|_{\H_i}\le 1$ for all $j$ and $i$.
\end{lemma}
The proof is standard.
\Cref{lem:jordan} allows us to state the following corollary.
\begin{corollary}
	\label{cor:tech}
	Let $\Pi_0$ and $\Pi_1$ be projectors with $\rank \Pi_0=\rank \Pi_1=d/2$, and $\rank(\Pi_0+\Pi_1)=d$ is full rank.
	Then there exists a unitary $\op V$ such that $\op V^\dagger\Pi_0\op Vr=\1_{d/2}\oplus0$.
	Furthermore,
	\begin{equation*}
	\op V^\dagger\Pi_1\op V=
	\begin{pmatrix}
	\op M_{aa} & \op M_{ab} \\ \op M_{ba} & \op M_{bb}
	\end{pmatrix},
	\end{equation*}
	where each $d/2\times d/2$ block $\op M_{ij}$ is diagonal and full rank, and $\op M_{ab}=\op M_{ba}<0$.
\end{corollary}
\begin{proof}
	Applying \cref{lem:jordan} to $\Pi_0$ and $\Pi_1$, we know that there exists a basis in which $\Pi_0$ is diagonal, and $\Pi_1=\bigoplus_i\op M_i$, where the $\op M_i$ are $2\times 2$ or $1\times 1$ Hermitian matrices.
	Re-order the basis again such that $\Pi_0=\1_r\oplus0$ with $r=d/2$; we denote the unitary transformation from Jordan's lemma with this  reordering as $\op V$.
	
	Under the same re-ordering, the matrix $\op V^\dagger\Pi_1\op V$ then necessarily has the form
	\begin{equation}
	\label{eq:matrix} \op V^\dagger\Pi_1\op V=
	\begin{pmatrix}
	\op M_{aa} & \op M_{ab} \\ \op M_{ba} & \op M_{bb}
	\end{pmatrix}
	=
	\begin{tikzpicture}
	[baseline] \node at (0,0) {$\begin{pmatrix}
		\begin{array}
		{cccc@{\hskip 8mm}cccc} \lambda_1 & 0         & \cdots & 0         & \xi_1  & 0      & \cdots & 0      \\
		0                                 & \lambda_2 & \ddots & \vdots    & 0      & \xi_2  & \ddots & \vdots \\
		\vdots                            & \ddots    & \ddots & 0         & \vdots & \ddots & \ddots & 0      \\
		0                                 & \cdots    & 0      & \lambda_r & 0      & \cdots & 0      & \xi_r  \\[5mm]
		\xi_1^*                           & 0         & \cdots & 0         & \mu_1  & 0      & \cdots & 0      \\
		0                                 & \xi_2^*   & \ddots & \vdots    & 0      & \mu_2  & \ddots & \vdots \\
		\vdots                            & \ddots    & \ddots & 0         & \vdots & \ddots & \ddots & 0      \\
		0                                 & \cdots    & 0      & \xi_r^*   & 0      & \cdots & 0      & \mu_r
		\end{array} \end{pmatrix}$};
	\draw[dashed,gray] (3.1,0) -- (-3.1,0) (0,-2.4) -- (0,2.4);
	\end{tikzpicture}
	\end{equation}
	A further global phase transformation (adjoining with a diagonal unitary, which leaves $\Pi_0$ invariant) allows us to assume that the $\xi_i=-|\xi_i|$, and the claim follows.
\end{proof}

We still need to show that the resulting Hamiltonian is in particular a graph Laplacian connecting the $T$\textsuperscript{th} entry in every block of $\Delta$ in \cref{eq:A} with the $T$\textsuperscript{th} entry in at least one block of $\Delta'$, which is where the input penalty terms are located.
Naturally, this will depend on the encoded circuit, and---using notation from the last proof---on the respective ranks of $\op M_{aa}$ and $\op M_{bb}$.
For this, we state the following lemma.

\begin{lemma}
	\label{lem:fullrank}
	Let $\op U$ encode a \no instance.
	Then $\Pin+\op U^\dagger\Pout\op U$ has full rank.
\end{lemma}
\begin{proof}
	It is easy to see that if $\ker(\Pin+\op U^\dagger\Pout\op U)\neq0$, there would be a state that can be accepted with perfect probability---contradicting the assumption that $\op U$ encodes a \no instance.
\end{proof}

\begin{figure}[t]
	\centering
	\newcommand{\ux}{1.4 cm}
	\begin{tikzpicture}
		\matrix (magic) [
			matrix of math nodes,
			nodes = {
					minimum size = \ux
				},
			left delimiter=(,
			right delimiter=),
			ampersand replacement=\&
		]
		{
			\Delta \&  \&  \& \& \& \&\& \\
			\& \Delta \& \&  \& \& \&\&  \\
			\&  \& \Delta \& \& \& \&\&  \\
			\& \& \& \Delta \& \& \& \&  \\
			\& \& \& \& \Delta \& \&\&   \\
			\& \& \&\& \& \Delta  \&     \\
			\& \& \&\& \& \& \Delta \&   \\
			\& \& \&\& \& \& \&\Delta    \\
		};
		\begin{scope}[
				xshift=-4*\ux,
				yshift=4*\ux,
				x=\ux,
				y=-\ux
			]
			\foreach \i in {1,...,7} {
					\draw[draw=black!20] (\i, 0) -- (\i, 8) (0, \i) -- (8, \i);
				};
			\draw (0, 4) -- (4, 4) -- (4, 8) (4, 0) -- (4, 4) -- (8, 4);

			\foreach \i in {1,...,4} {
					\node[scale=.66,blue] at (\i-1+.07, \i-1+.1) {$1$};

					\node[scale=.66,red] at (\i -.1, \i -.1) {$\lambda_{\i}$};
					\node[scale=.66,red] at (\i + 4 -.1, \i + 4 -.1) {$\mu_{\i}$};
					\node[scale=.66,red] at (\i -.2, \i + 4 -.1) {$-|\xi_{\i}|$};
					\node[scale=.66,red] at (\i + 4 -.2, \i -.1) {$-|\xi_{\i}|$};

					\draw[red!40] (\i, \i - .1) -- (\i + 4 - .4, \i - .1) (\i + 4 - .1, \i) -- (\i + 4 - .1, \i + 4 - .2);
					\draw[red!40] (\i - .1, \i) -- (\i - .1, \i + 4 - .2) (\i, \i + 4 - .1) -- (\i + 4 - .2, \i + 4 - .1);
				};

			\draw[->|] (-.5,.3) -- (-.5,0);
			\draw[->|] (-.5,.7) -- (-.5,1);
			\node at (-.5,.5) {$T$};
			\draw[->|] (-.8,3.8) -- (-.8,0);
			\draw[->|] (-.8,4.2) -- (-.8,8);
			\node at (-.8, 4) {$d$ blocks};
			\draw[->|] (8.5,1.6) -- (8.5,0);
			\draw[->|] (8.5,2.4) -- (8.5,4);
			\node at (8.5 + .3, 2) {\shortstack[l]{$\rank\Pin$\\blocks}};
		\end{scope}
	\end{tikzpicture}
	\caption{Sketch of matrix $(\1\otimes\op V)^\dagger(\Delta\otimes\1 + \ketbra1\otimes\Pin +\ketbra T\otimes M)(\1\otimes\op V)$.}
	\label{fig:matrix}
\end{figure}

This technical machinery allows us to prove \cref{theo:kitaev} in the next section.

\subsubsection{An $\Omega(T^{-2}$) Promise Gap for Kitaev's Hamiltonian}
Remember that we have to show that whenever the encoded circuit $\op U$ encodes a \no instance, then $\lmin(\Hkit)=\Omega(T^{-2})$.
\begin{proofsomething}{theo:kitaev} We remind the reader of the precise claim:
\kitaev*
\begin{proof}
	Since the circuit is a \no-instance, there exists a minimal acceptance probability for any (valid) input $\ket x\in\ker\Pin$, which we denote with $\epsilon$.
	By \cref{lem:fullrank}, we know that we can apply \cref{cor:tech} to $\Pin$ and $\op U^\dagger\Pout\op U$;  the unitary transformation bringing the latter into a form \cref{eq:matrix} while leaving $\Pin=\1\oplus 0$ we will denote with $\op V$.

	Now perform this similarity transform on \cref{eq:hamKitaev}, i.e.
	\begin{equation*}
		(\1\otimes\op V^\dagger)\op W^\dagger\Hkit\op W(\1\otimes\op V)= \Delta\otimes 1+\proj 0\otimes (\1\oplus 0)+\proj T\otimes (\op V^\dagger\op U^\dagger\Pout\op U\op V).
	\end{equation*}
	The resulting matrix is depicted in \cref{fig:matrix}.
	Since $\op U$ encodes a quantum circuit, we can immediately calculate the magnitudes of the $\lambda_i,\mu_i$ and $\xi_i$'s (see \cref{eq:matrix,fig:matrix} as a reference).

	\begin{enumerate}
		\item If $\ket x\in\field C^d$, then $\bra x\op U^\dagger\Pout\op U\ket x=\|\Pout\op U\ket x\|^2$ is the probability that the circuit $\op U$ rejects input $\ket x$.
		\item If $\ket x\in\ker\Pin$---i.e.\ for the $\op M_{bb}$ block---$\bra x\op U^\dagger\Pout\op U\ket x\ge1-\epsilon$.
		      Otherwise we have no non-trivial lower bound.
		\item Assuming $\dim\supp\Pout=d/2$ (discard the rest), we can thus immediately conclude that each matrix block
		      \begin{equation*}
			      \op P_i =
			      \begin{pmatrix}
				      \lambda_i & -|\xi_i| \\ -|\xi_i| & \mu_i
			      \end{pmatrix}
		      \end{equation*}
		      is of $\rank\op P_i=1$ with $1\ge\mu_i\ge 1-\epsilon$.
		\item Since the set of eigenvalues of a matrix is unique, we could continue diagonalizing $\op V^\dagger\Pout\op V$ by further diagonalizing each block $\op P_i$ separately; the result would be the same overall matrix as if we had diagonalized $\Pout$ in one step.
		      Since $\Pout$ is a psd projector, we thus know that each of the $\op P_i$ has to be a psd projector, and we can conclude
		      \begin{equation*}
			      \op P_i =
			      \begin{pmatrix}
				      \eta_i^2\mu_i & -\eta_i\mu_i \\ -\eta_i\mu_i & \mu_i
			      \end{pmatrix}
		      \end{equation*} for some $\eta_i:=\lambda_i/\mu_i\ge0$.
		      Then $1=\tr\op P_i=\mu_i(1+\eta_i^2) \ge (1-\epsilon)(1+\eta_i^2)$, and thus
		      \begin{equation*}
			      \eta_i \le \sqrt{\frac{1}{1-\epsilon}-1} = \sqrt{\frac{\epsilon}{1-\epsilon}} \le \sqrt{\frac\epsilon2},
		      \end{equation*}
		      and therefore $\lambda_i =\eta_i^2\mu_i \le \epsilon\mu_i/2 \le \epsilon/2$, and $|\xi_i|=\eta_i\mu_i \le \sqrt{\epsilon/2}$.
		\item From $\mu_i\le1$, one can obtain a trivial bound $1\le1+\eta_i^2$, and thus obviously $\lambda_i\ge0$.
	\end{enumerate}

	\noindent
	What remains to be analyzed is the spectrum of each of the pairs of blocks in \cref{fig:matrix}, i.e.\ blocks of the form $\Delta+\Delta'+$ matrix $\op P_i$ spread to the submatrix indexed by $T,2T$; more precisely, each block will have twice the size of $\Delta$, and we can write it as a stoquastic matrix on $\field C^2\otimes\field C^{T+1}$:
	\begin{equation}
		\label{eq:stoq-double-delta}
		\op S = (\ketbra 0 + \ketbra 1)\otimes\Delta + \ketbra 0\otimes\ketbra 0 + \underbrace{\mu\left[\ketbra 1 + \eta^2 \ketbra 0 - \eta (\ketbra{0}{1}  +\ketbra{1}{0})\right]}_{=:\op L}\otimes\ketbra T,
	\end{equation}
	where $\mu\ge 1-\epsilon$ and $\eta\le\sqrt{\epsilon/2}$.
	For e.g.\ $\dim \op S=8$ and reversing the second half of the time register (i.e.\ reversing the basis from $T$ to $2T$), this matrix looks like
	\begin{equation*}
		\op S_8=
		\begin{pmatrix}
			2 & -1 & 0 & 0 & 0 & 0 & 0 & 0 \\ -1 & 2  & -1 &      0       &    0     & 0  & 0  & 0  \\ 0  & -1 & 2  &      -1      &    0     & 0  & 0  & 0  \\ 0  & 0  & -1 & \mu\eta ^2+1 & -\mu\eta & 0  & 0  & 0  \\ 0  & 0  & 0  &   -\mu\eta   &  1+\mu   & -1 & 0  & 0  \\ 0  & 0  & 0  &      0       &    -1    & 2  & -1 & 0  \\ 0  & 0  & 0  &      0       &    0     & -1 & 2  & -1 \\ 0  & 0  & 0  &      0       &    0     & 0  & -1 & 1
		\end{pmatrix}
		.
	\end{equation*}
	Using a variational argument with the state
	\begin{equation*}
		\ket\phi:=\sum_{i=1}^T\sin\left(\frac{\pi i}{2T}\right)(\ket i + \ket{i+T})
	\end{equation*}
	and \cref{lem:lb}, we can explicitly show that $\op S=\Omega(T^{-2})$, independent of $\mu$ and $\eta$.
	The bound can be evaluated using five cases, which we summarize below.
	\begin{align*}
		\intertext{$t=1$:}
		\left(2 \sin \left(\frac{\pi }{2 T}\right)-\sin \left(\frac{\pi }{T}\right)\right) \csc \left(\frac{\pi }{2 T}\right) & =\frac{\pi ^2}{4 T^2}+O(T^{-3})                                                       \\ \intertext{$1< t\le T-1$ and $T+1<t<T$:}
		4 \sin ^2\left(\frac{\pi }{4 T}\right)                                                                                & =\frac{\pi ^2}{4 T^2}+O(T^{-3})                                                       \\ \intertext{$t=T$:}
		\eta ^2 \mu -\eta  \mu  \sin \left(\frac{\pi }{2 T}\right)-\cos \left(\frac{\pi }{2 T}\right)+1                       & =\eta ^2 \mu -\frac{\pi  \eta  \mu }{2 T}+\frac{\pi ^2}{8 T^2}+O(T^{-3})              \\ \intertext{$t=T+1$:}
		\mu -\eta  \mu  \csc \left(\frac{\pi }{2 T}\right)-2 \cos \left(\frac{\pi }{2 T}\right)+1                             & =-\frac{2 T (\eta  \mu )}{\pi }+(\mu -1)-\frac{\pi  \eta  \mu }{12 T}+\frac{\pi ^2}{4
			T^2}+O(T^{-3})                                                                                                                                                                                                \\ \intertext{$t=2T$:}
		1-\sin \left(\frac{\pi  (T-1)}{2 T}\right)                                                                            & =\frac{\pi ^2}{8 T^2}+O(T^{-3})
	\end{align*}
The claim of \cref{theo:kitaev} follows.
\end{proof}
\end{proofsomething}

For convenience, we further collect the findings regarding the block matrix decomposition of the full circuit Hamiltonian given in the proof of \cref{theo:kitaev} in the following lemma.
\begin{lemma}
	Kitaev's construction can be block-decomposed as $\Hkit=\bigoplus \op S_i$, where each block $\op S_i$ has the form given in \cref{eq:stoq-double-delta}.
	The ground state energy of each of these blocks is $\Omega(T^{-2})$, and thus $E_p(\Hkit)=\Omega(T^{-2})$.
\end{lemma}

\section{Limitations on Further Improvement}\label{sec:limitations}
\subsection{Linear Clock Registers with Endpoint Penalties}\label{sec:limitation-linear}
The proof of \cref{theo:tridiag} is based on applying a sharp spectral gap bound for birth-and-death Markov chains to a quantum-to-classical mapping that has been studied previously in the closely related context of universal adiabatic computation \cite{aharonov2008adiabatic} and the complexity of stoquastic Hamiltonians~\cite{bravyi2006merlin,bravyi2009complexity}.
A new feature of our application is the realization that this quantum-to-classical mapping defines a Markov chain even for tridiagonal Hamiltonian matrices with arbitrary complex entries, while previous applications have been restricted to cases for which $\op H$ has all non-positive off-diagonal matrix entries in the time register basis.

\begin{proofsomething}{theo:tridiag}
In this section we continue with the notation of \cref{eq:ham}, but now we use the freedom to shift the overall energy to set $a_t \geq 0$ for all $t$, such that the ground state energy $E$ satisfies $0 \leq E < 1$; this allows us to show optimality for tri-diagonal clock constructions, re-stated in the following theorem.
\tridiag*
\begin{proof}
Define  $\op G := (\1-\op H)/(1-E)$ to be a shifted and rescaled version of $\op H$ which is designed to satisfy $\op G \ket{\psi} = \ket{\psi}$, where $\ket\psi$ labels the ground state of $\op H$.
For all $t,t'\in \{0,\ldots,T\}$, define
\[
	\op P_{t,t'} :=
	\begin{cases}
		\psi_{t'} \op G_{t,t'} \psi_{t}^{-1} & \text{if $\psi_t \neq 0$ and $\psi_{t'} \neq 0$}, \\
		0                                    & \text{otherwise}.
	\end{cases}
\]
In the following lemma we will show that the $\op P_{t,t'}$ are transition probabilities for an irreducible Markov chain on $\{0,\ldots,T\}$, i.e.\ in particular that they are all nonnegative.
First, observe that if $\op H$ is stoquastic as in previous applications, then $\op G$ is a nonnegative matrix in the time register basis, $\psi$ has nonnegative amplitudes in this basis by the Perron-Frobenius theorem, and so $\op P_{t,t'}$ is explicitly nonnegative.
Here we show that even when $\op G$ contains arbitrary complex matrix entries (while being tridiagonal) we still have $\op P_{t,t'}\geq 0$, because of cancellations that occur between the matrix elements of $\op G$ and the amplitudes of the ground state wave function in the time register basis.
Continuing with the same notation used in \cref{eq:clockham}, we state the following lemma.
\begin{lemma}
	\label{lem:unfrustrated} If $\psi_0 \neq 0$, $\psi_T \neq 0$, and $b_t \neq 0$ for $t = 0,\ldots,T$, then $\psi_t \neq 0$ for $t = 1,\ldots,T-1$ and $\op P_{t,t+1}  =\psi_{t+1} \op G_{t,t+1} \psi_t^{-1} \geq 0$ for all $t \in \{0,\ldots,T-1\}$.
\end{lemma}
Before proving the lemma, note that the conditions may be taken to hold without loss of generality, since $\psi_0 = 0$ or $\psi_T = 0$ immediately implies \cref{theo:tridiag}l. Similarly, if $b_{t'} = 0$ for some $t'$, then  $\ket{\psi^\perp} := \sum_{t=0}^T \psi^\perp_t \ket{t}$ defined by
\begin{equation}
	\label{eq:perp} \psi^\perp_t :=
	\begin{dcases}
		\phantom{+}\frac{\psi_t}{\psi^2([0,t'])}       & 0 \leq t \leq t'  \\
		 -\frac{\psi_t}{\psi^2([t'+1,T])}  & t' < t \leq T - 1
	\end{dcases}
\end{equation}
satisfies $\braket{ \psi^\perp}{\psi} = 0$ and $\op H \ket{\psi^\perp} = E\ket{\psi^\perp}$, which implies $\Delta_{\op H} = 0$---so again  \cref{theo:tridiag} holds in this case---note that the same idea behind \cref{eq:perp} can be used to upper bound the spectral gap by the minimum of $\pi_t \pi_{t+1}$ over all $t\in \{0,\ldots,T-1\}$, such that $\psi^2([0,t'])$ and $\psi^2([t'+1,T])$ are both $\Omega(1)$.

Now turning to the proof of \cref{lem:unfrustrated}.
From $\op H \ket{\psi} = E \ket{\psi}$ we have that
\begin{align}
	 & a_{0}\psi_0 + b_{0}\psi_1 = E \psi_0, \label{eq:left-amplitude}                                                                       \\
	 & b^*_{i-1}\psi_{i-1} + a_i \psi_i +  b_i\psi_{i+1} = E \psi_i \quad \text{for } i = 1,\ldots,T-1\text{, and} \label{eq:mid-amplitudes} \\
	 & a_{T} \psi_{T} + b^*_{T-1} \psi_{T-1} = E \psi_{T}. \label{eq:right-amplitudes}
\end{align}
Since $\op P_{t,t'} = 0$ when $|t - t'| > 1$, our goal is to show $\op P_{t,t+1} > 0$ for $t\in \{0,\ldots,T-1\}$, and $\op P_{t,t-1}>0$ for $t\in\{ 1,\ldots,T\}$.
The first claim $\op P_{t,t+1} > 0$ will follow by showing that $E$ is minimized when $\psi_t \neq 0$ and $\psi_{t+1} b_t \psi^{-1}_t < 0$ for all $t= 0,\ldots,T-1$.
The second claim $\op P_{t-1,t} > 0$ is then implied immediately, since
\[
	\psi_t b^*_{t} \psi_{t+1}^{-1} = \left(\frac{|\psi_t|}{|\psi_{t+1}|} \right)^2 \frac{\psi^*_{t+1}}{\psi^*_t} b^*_t =  \left(\frac{|\psi_t|}{|\psi_{t+1}|} \right)^2 \left(\psi_{t+1}b_t\psi_t^{-1}  \right)^*.
\]

Rearranging \cref{eq:left-amplitude} yields $\psi_1 b_{0}\psi^{-1}_0 = E - a_0$, and since $E - a_{0}$ is real the value of $E$ implied by this equation alone is minimized when the LHS is negative.
This observation will be taken as the base case for an inductive argument over the finite set $\{1,\ldots,T-1\}$.

The inductive hypothesis is that the value of $E$ implied by considering only equations $0$ through $t$ in the list \cref{eq:mid-amplitudes} is minimized when $\psi_{t}b_{t-1}\psi^{-1}_{t-1}$ is negative for $1,\ldots,t$, and this will be used to show that the minimum value of $E$ that satisfies equations $0$ through $t+1$ in \cref{eq:mid-amplitudes} will be achieved when $\psi_{t+1}b_{t}\psi^{-1}_{t}$ is negative as well.
Using the fact that $\psi_{t} \neq 0$ from the inductive hypothesis, we may express \cref{eq:mid-amplitudes} as
\[
	b^*_{t-1}\frac{\psi_{t-1}}{\psi_t} +  b_{t} \frac{\psi_{t+1}}{\psi_i} = E - a_{t}  \quad \text{for } t = 1,\ldots,t.
\]
Since $E - a_t$ is real and $\psi_{t-1} b^*_{t-1} \psi^{-1}_t = (|\psi_{t-1}|/|\psi_t|)^2 (\psi_{t}b_{t-1}\psi^{-1}_{t-1})^*$ is negative by the inductive hypothesis, the value of $E$ implied by equations 0 through $t+1$ in the list \cref{eq:mid-amplitudes} will indeed be minimized by taking $\psi_{t+1} b_{t}\psi^{-1}_t$ to be negative.
This establishes the inductive claim;
a simple argument shows the same for \cref{eq:right-amplitudes}, which completes the proof of \cref{lem:unfrustrated}.

Having established that $\op P_{t,t'} \geq 0$ we now list several standard facts which have been previously applied to $\op P$ when $\op G$ is nonnegative, which can also be seen in the present case by direct computation:
\begin{enumerate}
	\item $\op P$ is a stochastic matrix, i.e.\ $\sum_{t' = 0}^T \op P_{t,t'} = 1$ for all $t \in \{0,\ldots,T\}$, and therefore it can be regarded as the transition matrix of a discrete time Markov chain.
	\item The largest eigenvalue of $\op P$ is equal to 1 and it corresponds to the unique principal eigenvector $\ket{\pi} = \sum_{t=0}^T |\psi_t|^2 \ket{t}$.
	      The probability distribution $\pi_t := \braket{ t }{\pi}$ is the stationary distribution of the corresponding Markov chain.
	\item The Markov chain defined by $\op P$ is reversible with respect to its stationary distribution,
	      \begin{equation*}
		      \pi_t \op P_{t,t'} = \braket{t'}{\psi} \braket{\psi}{t} \op G_{t,t'} = \left (\braket{t}{\psi} \braket{\psi}{t'} \op G^*_{t,t'}\right )^* = \pi_t' \op P_{t',t}.
	      \end{equation*}
	\item If $\ket{\psi_0} , \ket{\psi_1} , \ldots , \ket{\psi_T}$ are the eigenvectors of $\op H$ with eigenvalues $E_0 < E_1 \leq \ldots \leq E_T$, then $\ket{\phi_k} := \sum_{x\in \Omega} \braket{\psi_0}{x} \braket{x}{\psi_k} \ket x$ is an eigenvector of $\op P$ with eigenvalue $( 1 - E_k)/(1- E_0)$.
	      Since this is the complete list of eigenvectors of $\op P$, we have shown that the spectral gaps of $\op H$ and $\op P$ satisfy
	      \begin{equation}
		      \Delta_{\op P} = (1-E) \Delta_{\op H}.
		      \label{eq:gap-relate}
	      \end{equation}
\end{enumerate}

The relation \cref{eq:gap-relate} means that we can apply techniques developed for upper bounding the spectral gap of Markov chains to the problem of upper bounding the spectral gap of $\op H$.
A non-trivial example of such an upper bound is \cref{eq:cheeger}: if the overlap of the stationary distribution with the end points $\ket{0}$ and $\ket{T}$ is constant, we can immediately conclude that the conductance $\Phi$ is $\BigO(T^{-1})$, by the fact that the stationary distribution is normalized; this implies $\Delta_{\op H}=\BigO(T^{-1})$.

It turns out we can obtain an even tighter bound by using a characterization of spectral gaps that applies specifically to birth-and-death chains~\cite{Chen2013}, which we state here as a lemma.
\begin{lemma}
	If $\op P$ is a birth and death chain with stationary distribution $\pi$, then the spectral gap $\Delta_{\op P}$ satisfies
	\begin{equation}
		\frac{1}{2 \ell} \leq \Delta_{\op P} \leq \frac{4}{\ell} \label{eq:ell-gap}
	\end{equation}
	where
	\begin{equation}
		\ell := \max \left \{ \max_{j : j \leq i'} \sum_{k = j}^{i' -1} \frac{\pi([0,j]}{\pi(k)\op P_{k,k+1}} , \max_{j: j > i'} \sum_{k=i'+1}^j \frac{\pi([j,T])}{\pi(k)\op P_{k,k-1}} \right\}\label{eq:bd-lemma}
	\end{equation}
	where $i'$ satisfies $\pi([0,i'])\geq 1/2$ and $\pi([i',n]) \geq 1/2$.
\end{lemma}
In the present case we are seeking a lower bound on $\ell$ in order to have an upper bound on the gap.
To simplify the formulas we assume that the stationary distribution of the weighted history state is symmetric around $t = T/2$ (otherwise the problem divides into two similar cases).
Since we are seeking a lower bound on $\ell$, we can ignore the factor of $\op P_{k,k+1} \leq 1$ in the denominator, and we are also free to replace the maximization over $j$ with any fixed choice of $j$.
With these simplifications and the choice of $j = 1$, \cref{eq:bd-lemma} becomes

\begin{equation*}
	\ell \geq  \psi^2_0\sum_{t = 1}^{T/2 - 1} \frac{1}{\psi^2_t }.
\end{equation*}
Applying the inequality of the arithmetic and geometric means yields
\[
	\sum_{t = 1}^{T/2 - 1} \frac{1}{\psi^2_t }  \geq \left( \frac{T}{2} - 1 \right) \left(\psi_1^2 \cdots \psi^2_{T/2-1}  \right)^{-1/k} \geq \left( \frac{T}{2} - 1 \right)^2 \left(\sum_{t=1}^{T/2-1} \psi_t  \right)^{-1},
\]
and so $\ell=\Omega(\psi^2_0 T^2)$.
Together with \cref{eq:ell-gap} we have that the spectral gap $\Delta_{\op P}$ is $\BigO(\ell^{-1})$, and it follows that $\Delta_{\op H} \cdot \psi^2_0=\BigO(T^{-2})$, as claimed.
\end{proof}
\end{proofsomething}

Finally, we note that \cref{theo:tridiag} can be interpreted as proving that the standard universal adiabatic construction plus the weighted endpoint modification made above is in a sense optimal for Hamiltonians of the form \cref{eq:ham}.
First, the problem of upper bounding the spectral gap of universal adiabatic constructions was addressed before~\cite{ganti2013gap} by combining the quantum lower bound for unstructured search with the technique of spectral gap amplification.
This previous work found a general $\tilde{\BigO}(T^{-1})$ bound (where the tilde hides logarithmic factors) on the spectral gap of \emph{any} adiabatic Hamiltonian, an $\tilde{\BigO}(T^{-2})$ gap for any frustration-free adiabatic Hamiltonian, and finally an $\tilde{\BigO}(T^{-2})$ bound on the spectral gap of modified Feynman Hamiltonians of the form \cref{eq:ham} when the weights near the endpoints satisfy a reasonable assumption for any adiabatic computation.
Our \cref{theo:tridiag} corroborates this last result by showing a tight $\BigO(T^{-2})$ upper bound on the spectral gap and the minimum overlap of the weighted history state with either endpoint of the computation.

\subsection{Diameter Bounds for simple generalized ULGs}\label{sec:sc-ulg-diameter-bound}
Following \cref{theo:tridiag}, an immediate question is whether we can overcome the upper bound on the product of penalty overlap and spectral gap, e.g.\ by deviating from the tridiagonal ``path'' evolution.
For instance, one might hope that increasing the connectivity of the ULG underlying the circuit Hamiltonian would yield a larger gap, or would allow us to circumvent \cref{theo:tridiag} altogether.
In this section, we will argue for a lower bound on the diameter of any such graph designed to increase the connectivity.

Assume we have a circuit Hamiltonian consisting of a family of simple ULGs \cite{Bausch2016} (as described in \cref{sec:limitationsIntro}) and a family of circuits $(C_i)_{i\in\field N}$, e.g.~a uniform verifier circuit family.
Further suppose that each of these circuits is ``optimally compiled'', in the sense that the product of local unitaries along any length $k$ path through the ULG cannot be compiled into fewer than $k$ local gates.  

This family of circuits can be implemented by a family of standard circuit Hamiltonians $\{\op H_i\}_{i \in \mathbb{N}}$, each of which has a linear clock ranging from $t=1,\ldots,T_i$, where $T_i:=|C_i|$ is the number of gates in the circuit $C_i$.
This corresponds to a path-like clock:
\begin{equation*}
	\includegraphics{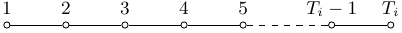}
\end{equation*}
Now assume there exists another Hamiltonian $\op H'$ with at least one path between $t=1$ and $t=T_i$, but of shorter length $T'<T_i$, e.g.
\begin{equation*}
	\includegraphics{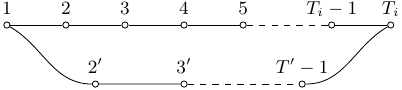}
\end{equation*}
\newcommand{\Hrest}{\op H_\textnormal{rest}}
Since $(1,\ldots,T_i,T'-1,\ldots,2',1)$ forms a loop, the assumption that the ULG is simple implies $(\op U_{1,2}\ldots \op U_{T_i -1,T_i})(\op U_{T' -1, T'-2}\ldots U_{2',1}) = \1$.  But demanding this condition for all the circuits in the family is inconsistent with the hypothesis that the circuit was compiled optimally in first place.

This demonstrates that if we want to increase the UNSAT penalty with more general clock graphs, we better use frustration, i.e.\ go beyond simple ULGs;
we discuss this case in \cref{sec:non-simple-ulg}.

\subsection{A Clock Diameter Bound}\label{sec:diameter-bound}
In \cref{sec:sc-ulg-diameter-bound}, we saw how in a generalized ULG, we need paths of a minimum length to encode computation; anything shorter will necessarily conflict with the computational output and introduce an additional error in a \yes case;
this, however, is a problem, since e.g.\ for \QMA-hardness proofs it is crucial to be able to amplify the \yes error to $\epsilon=\smallO(\text{promise gap})$.

But if we know that there have to be straight sections of length $T$ within any graph, how does this graph \emph{diameter} correspond to the spectral gap of the clock Hamiltonian?
And can we combine this with the case of non-simple ULGs (\cref{sec:non-simple-ulg})?
By \cref{lem:geo}, we know that $E_p(\epsilon,T)\le E(\Hclock(T)+\op P)$, even for more general positive semi-definite penalties on an arbitrary number of vertices, all of which we collect in the term $\op P$.
In particular, this tells us that we can obtain an upper bound on the UNSAT penalty solely by looking at the clock part of the Hamiltonian.

\begin{proofsomething}{th:diam}
The most general such clock is an arbitrary Hermitian matrix, and we develop a diameter bound on it following \cite{KelnerMITLecturenotes2009}.
This proves the last of our theorems, i.e.
\thdiam*
\begin{proof}
Let $\op H$ be Hermitian on $\field C^T$, and---as in \cref{sec:limitation-linear}---define
\begin{equation}\label{eq:G}
\op G:=(\op H-\lmin)/(\| \op H \| - \lmin).
\end{equation}
Then $\spec(\op G)\subset[0,1]$.
For two states $\ket u,\ket v\in\field C^T$, we define their \emph{distance} under $\op G$ as
\[
	\dist_{\op G}(\ket u,\ket v):=\min\{ k : \bra u \op G^k \ket v \neq 0\}
\]
if it exists, and we set $\dist_{\op G}(\ket u,\ket v)=\infty$ otherwise; this is the case e.g.\ if $\op G$ is block-diagonal, and $\ket u$ and $\ket v$ are vectors with support completely contained in disjoint blocks.

The reason for calling it a \emph{distance} is obvious from the equality
\begin{equation*}
	\dist_{\op G}(\ket u,\ket v)=\min\left\{d: \bra u\prod_{i=1}^d\left(\sum_{jk}\proj j\op G\proj k\right)\ket v\neq 0\right\},
\end{equation*}
i.e.\ it catches the minimum number of nonzero jumps under $\op G$ necessary to reach some overlap between $\ket u$ and $\ket v$.
The \emph{diameter} of $\op G$ is then defined as
\[
	\diam\op G:=\max_u\min_v\{\dist_{\op G}(\ket u,\ket v)\}.
\]
It is clear that $\diam\op G=\diam\op H$.

If we use the spectral decomposition $\op G=\sum_i\lambda_i\proj{\psi_i}$ and some polynomial $p$, then for any states $\ket u$ and $\ket v$, we have \newcommand{\pmin}{\pi_{\mathrm{min}}}
\begin{align*}
	|\bra u p(\op G)\ket v| & =\left|\sum_i p(\lambda_i)\braket{u}{\psi_i}\braket{\psi_i}{v}\right| \\ &\ge p(1)|\pi_{u,v}| + \left|\sum_{i\ge 2}p(\lambda_i) u_i \bar v_i\right|\\ &\ge p(1)|\pi_{u,v}| - \max_{i\ge 2}p(\lambda_i)\sum_{i\ge 2}|u_i||v_i|\\ &\ge p(1)|\pi_{u,v}| - \max_{i\ge 2}p(\lambda_i),
\end{align*}
where $\pi_{u,v}:=u_1\bar v_1$, and $u_i:=\braket{u}{\psi_i}$, $\bar v_i:=\braket{\psi_i}{v}$.

One can show that there exists a family of polynomials $p_k$ with $\deg p_k=k$, $p_k(1)=1$ and $|p_k(x)|\le 2(1+\sqrt{2\Delta})^{-k}\ \forall x\in[\Delta,1]$.
Assume for now that $\pi_{u,v}>0$.
We want to find a condition under which $p(\lambda_i)<\pi_{u,v}$, so we resolve
\begin{equation*}
	2(1+\sqrt{2\Delta})^{-k}<\pi_{u,v} \quad\Leftrightarrow\quad -k\ln(1+\sqrt{2\Delta})<\ln\left(\frac{\pi_{u,v}}{2}\right) \quad\Leftrightarrow\quad k>\frac{\ln(2/\pi_{u,v})}{\ln(1+\sqrt{2\Delta}}.
\end{equation*}
For $x\in(0,1)$ we further have a uniform bound\footnote{This bound is tight up to a factor of $2\ln2$, which is reached at $x=1$.}\ \ of $\ln(1+\sqrt{2x})\le(1+1/\sqrt{x})$, which yields a diameter bound of
\[
	\diam\op G\le\left(1+\frac{1}{\sqrt{2\Delta}}\right)\ln(2/\pi_{u,v}).
\]
This inequality still depends on our choice of states $\ket u$ and $\ket v$.
Assuming that $\pmin=\min_i|\braket{i}{\psi_1}|^2$ denotes the minimum ground state overlap with any basis state $\ket i$, we know that $|\pi_{u,v}|^2\ge\pmin$.
If we further assume $\pmin>0$, then the diameter bound reads
\[
	\diam\op G\le\left(1+\frac{1}{\sqrt{2\Delta}}\right)\ln(2/\pmin).
\]

What about the case where $\pmin=0$?
This would mean that there is at least one state $\ket x$ for which $\braket{x}{\psi_1}=0$.
Since we are interested in encoding computation in the ground state of a Hamiltonian, and $\op G$ serves as clock, we can dismiss those parts of the Hilbert space in which the ground states shows zero connectivity; this is the case if $\pmin>0$, and otherwise we restrict $\ket u, \ket v\in\supp \proj{\psi_1}$.
For the degenerate case we block-diagonalize $\op G$ and regard each such block separately.

To formalize this, we define a variant of the diameter which captures the ground space connectivity:
\begin{equation*}
	\diam' \op G:=\max_{u\in\supp\proj{\psi_1}}\min_{v\in\supp\proj{\psi_1}}\{\dist_{\op G}(\ket u,\ket v)\}.
\end{equation*}
The diameter bound then reads
\begin{equation}\label{eq:diameter-bound}
	\diam'\op G\le\left(1+\frac{1}{\sqrt{2\Delta}}\right)\ln(2/\pmin'),
\end{equation}
where $\pmin'=\min_{i\in\supp\proj{\psi_1}}|\braket{i}{\psi_1}|^2$.
Using the fact from \cref{sec:sc-ulg-diameter-bound} that the diameter for a given circuit of $T$ gates is at least $T$, and applying the definition of $\op G$ in terms of the Hamiltonian $\op H$, \cref{eq:G}, we obtain \cref{th:diam}.
\end{proof}
\end{proofsomething}

\subsection{Frustrated ULGs}\label{sec:non-simple-ulg}
In \cref{sec:sc-ulg-diameter-bound} we have seen that if we hope to increase the UNSAT penalty by going beyond a linear clock evolution, we will also have to go beyond simple ULGs.
The case left to be analyzed is therefore when we allow the quantum register to be frustrated, in the sense that there are two clock labels $a$ and $b$ with two distinct paths connecting them, and such that the product of unitaries along both paths is not identical.

As a simple example, consider the following ULG.
\begin{equation*}
	G=\raisebox{-.95cm}{\includegraphics[]{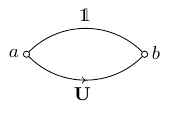}}
\end{equation*}
In case that $\op U=\sigma_x$, we can immediately write down the ULG Hamiltonian as
\begin{equation*}
	\op H_G=
	\begin{pmatrix}
		2  & 0  & -1 & -1 \\
		0  & 2  & -1 & -1 \\
		-1 & -1 & 2  & 0  \\
		-1 & -1 & 0  & 2
	\end{pmatrix}
	.
\end{equation*}
As we can see, this Hamiltonian by itself is already a graph Laplacian for the following graph:
\begin{equation*}
	\includegraphics[]{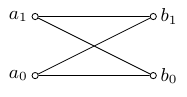}
\end{equation*}
However, our goal is to reduce the number of edges; if we instead work in the $\ket\pm$-eigenbasis of the off-diagonal blocks of $\op H_G$, $\1+\sigma_x$, then
\begin{equation*}
	\op H_G\sim
	\begin{pmatrix}
		2 & 0 & -2 & 0 \\
		0 & 2 & 0  & 0 \\
		2 & 0 & 2  & 0 \\
		0 & 0 & 0  & 2
	\end{pmatrix}
	,
\end{equation*}
which in turn corresponds to the graph
\begin{equation*}
	\includegraphics[]{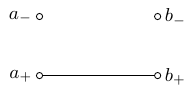}
\end{equation*}
but such that the vertices $a_-$ and $b_-$ carry an \emph{extra penalty} of $2$.

This argument also works for more general $\op U$: labelling the eigenvalues of $\1+\op U$ with $\lambda_1,\ldots,\lambda_n$, we have
\begin{equation}\label{eq:non-simple-ulg-2}
	\op H_G\sim
	\begin{pmatrix}
		2\1 & \diag(\lambda_1,\ldots,\lambda_n) \\ \diag(\bar{\lambda_1},\ldots,\bar{\lambda_n}) & 2\1
	\end{pmatrix}
	.
\end{equation}
The eigenvalues thus determine the strength of the edge connection; for $\lambda_i=0$ (as in the case of $\op U=\sigma_x$), the graph becomes completely disconnected; the adjoining vertices of the edge hence carry an extra penalty of
\begin{equation*}
	\op P=\sum_{i=1}^n (2-|\lambda_i|)\proj{\lambda_i},
\end{equation*}
where as in \cref{eq:non-simple-ulg-2} the $|\lambda_i|\le \|\op U+\1\|\le 2$.
Similar arguments can be derived for more complicated graphs;
as disconnected graph components introduce degeneracies, each component would need its own set of in- and output constraints in order to allow a distinction between \yes and \no case.

However, if that \emph{is} the case, then one can analyse the UNSAT penalty of each component separately.
As can be seen in \cref{def:unsat}, the \no-\yes energy difference stems from $\op H_\mathrm{FK}$---the Hamiltonian that includes in- and output penalties; $E(\Hprop)$ is static.
The UNSAT penalty would thus arise from the best such disconnected block, which we could have written down directly in first place.
It is thus highly doubtful that an explicitly-constructed non-simple ULGs will feature an amplified UNSAT penalty over the simple case.

\section{Outlook}\label{sec:outlook}
One of the main aims of the present work is to motivate new ideas in quantum ground state computation by focusing on the quantum UNSAT penalty as a metric for the improvement of circuit Hamiltonians.
We discuss a range of open problems related to the UNSAT penalty, some of which either appear more tractable, or lead to a different perspective on some of the open challenges facing this field; in particular, we are implicitly discussing \emph{relative} energy penalties that are not simply made larger by e.g.\ increasing the overall norm of the Hamiltonian.

\paragraph{The Classical Baseline.}
The classical Cook-Levin theorem encodes the history of a classical circuit into the satisfying assignment of a 3-SAT formula.
If the computation has $T$ time steps, then the associated constraint satisfaction problem has $\BigO(T)$ local terms.
If each f those has a constant norm, then the classical UNSAT penalty is $\Omega(1)$.
Therefore we ask: is it possible for a circuit Hamiltonian containing $\BigO(T)$ local terms of bounded norm---which may be of a form more general than \cref{eq:ham}---to achieve an UNSAT penalty that is independent of the length of the computation?

\paragraph{Macroscopic UNSAT penalty.}
Building on the previous question which asks whether the UNSAT penalty can be made independent of the length of the computation, we further ask whether the UNSAT penalty can be made to scale macroscopically with the number of qubits $n$ in the underlying many-body quantum system.
Specifically, is there a circuit Hamiltonian with $\BigO(\text{poly}(n) T)$ local terms that achieves a $\text{poly}(n)$ UNSAT penalty that is independent of $T$?
Such a construction could be a useful step towards fault-tolerant adiabatic computation.

An intuition for this connection can be gained by considering a construction for energetically encoded fault-tolerant classical computation, whereby each logical bit could be encoded as an arrangement of spins in a self-correcting model (e.g.\ the 2D Ising model).
The UNSAT penalty could then have a macroscopic scaling---i.e.\ with the number of physical spins representing each logical bit---that is independent of $T$.

\paragraph{Constant relative UNSAT penalty.}
A circuit Hamiltonian with $O(m)$ local terms of bounded norm, where $m = \poly(T)$, with constant relative UNSAT penalty $E_p / m$ would yield a proof of the quantum PCP conjecture by spectral gap amplification.
The reduction consists of applying the circuit Hamiltonian with constant relative UNSAT penalty to the circuit verifier that decides the ground state energy of the arbitrary input local Hamiltonian.

\vspace{.5cm}
It is a testament to Feynman's great legacy that an idea first introduced in 1985 has had such a profound impact on a remarkably wide scope of research, from condensed matter physics to quantum computation, and that despite the growth of the field of Hamiltonian complexity his original construction continues to remain essentially unchanged to date.
We do not know whether or where limitations of improving the circuit-to-Hamiltonian construction will be reached, but hope that our contribution marks some of its limits, and helps to push other boundaries a little further.

\section{Acknowledgements}
J.\,B.\ acknowledges support from the German National Academic Foundation, the EPSRC (grant 1600123), and the Draper's Research Fellowship at Pembroke College.
E.\,C. acknowledges support provided by the Institute for Quantum Information and Matter, an NSF Physics Frontiers Center (NSF Grant PHY-1125565) with support of the Gordon and Betty Moore Foundation (GBMF-12500028).
J.\,B.\ would also like to thank Thomas Vidick and the IQIM for hospitality during spring 2016 and summer 2017.
We thank Toby Cubitt for useful discussions regarding \cref{sec:jordan-proof,sec:direct-proof}.
\newpage
\printbibliography

\newpage
\appendix
\addtocontents{toc}{\protect\setcounter{tocdepth}{1}}
\section{Appendix}
\subsection{An  $\Omega(T^{-2})$ Scaling from Padding with Identities}\label{sec:padding-ham}
In this appendix we show that a modified version of Kitaev's circuit-to-Hamiltonian construction that pads the input and output clock states with an $\Omega(T)$-sized identity circuit \emph{and} spreading out the input and output penalty yields an $\Omega(T^{-2})$ UNSAT penalty that can be proven by the geometrical lemma.  Let $\op U_1,\ldots,\op U_T$ be a quantum circuit on some space $\mathcal H$.
Define a padded circuit Hamiltonian on $\field{C}^{2T}\otimes\mathcal H$ as
\begin{align}
	\op H & = \sum_{t=-T/2}^{0}(\proj t + \proj{t+1} - \ketbra{t+1}{t} - \ketbra{t}{t+1})\otimes\1 \\ &+ \sum_{t=0}^{T}(\proj t + \proj{t+1})\otimes\1 - \ketbra{t+1}{t}\otimes \op U_t - \ketbra{t}{t+1}\otimes \op U_t^\dagger \\ &+ \sum_{t=T}^{3T/2}(\proj t + \proj{t+1} - \ketbra{t+1}{t} - \ketbra{t}{t+1})\otimes\1.
\end{align}
The overall encoded circuit of this Hamiltonian is thus
\begin{equation*}
	\underbrace{\1\cdots\1}_{T/2\text{\ times}} \op U_T\cdots \op U_1  \underbrace{\1\cdots\1}_{T/2\text{\ times}} =: \op U_{3T/2}'\cdots \op U_{-T/2}'.
\end{equation*}
The ground state of this unbiased clock construction is a uniform superposition over the history of the computation, i.e.\ $\ket{\Psi}=\frac1{\sqrt{2T}}\sum_{t=-T/2}^{3T/2}\ket t\otimes \op U_t'\cdots \op U_{-T/2}'\ket{\phi}$.

For some input and output penalty $\Pin,\Pout$ acting on the circuit qubits, we define the padded penalty Hamiltonian as
\begin{equation}
	\op P = \sum_{t=-T/2}^{0}\proj t\otimes\Pin + \sum_{t=T}^{3T/2}\proj t\otimes\Pout.
\end{equation}
Denote the projector onto the kernel of this penalty term with
\begin{equation*}
	\Pi_\mathrm{pen}:= \sum_{t=-T/2}^0\proj t\otimes\Pin^\perp +\sum_{t=T}^{3T/2}\proj t\otimes\Pout^\perp +\sum_{t=1}^{T-1}\proj t\otimes\1.
\end{equation*}
Let $W$ be the usual diagonalizing operator for $\op H$.
We first bound the angle $\theta$ between $\ker\op  P$ and $\ker \op H$: 
\begin{align*}
	\cos^2\theta
	 &= \max_{\ket{\eta'}\in\ker \op W\op H\op W^\dagger} \bra{\eta'} \left( \sum_{t=-T/2}^0\proj t\otimes\Pin^\perp +\sum_{t=T}^{3T/2}\proj t\otimes\Pout^\perp +\sum_{t=1}^{T-1}\proj t\otimes\1 \right)\ket{\eta'} \\
	 &= \max_{\ket{\phi}} \frac 1{2T} \bra{\phi}\left( \sum_{t=-T/2}^0  \proj t\otimes\Pin^\perp +\sum_{t=T}^{3T/2}  \proj t\otimes\Pout^\perp +\sum_{t=1}^{T-1}  \proj t\otimes\1 \right) \ket{\phi} \\ &= \max_{\ket{\phi}} \frac 14 \bra{\phi} \Pi_\text{in} + \op U\Pi_\text{out} \op U^\dagger \ket{\phi} + \frac 12.
\end{align*}
The inner product is then bounded by the acceptance probability $\epsilon$ of the circuit, i.e.
\begin{align*}
	\max_{\ket{\phi}} \bra{\phi} \Pi_\text{in} + \op U\Pi_\text{out} \op U^\dagger \ket{\phi} \le + \cos\vartheta \le 1 + \epsilon.
\end{align*}
Overall, this gives a bound
\begin{equation}
	\cos^2\theta \le \frac14(3 + \epsilon),
\end{equation}
which gives precisely the same UNSAT penalty as the un-padded version used in \cref{sec:theo1proof}.

\subsection{Proof of $\Omega(T^{-2})$ Scaling without Jordan's Lemma}\label{sec:direct-proof}
We start in \cref{sec:jordan-proof}, after \cref{eq:A}.
If our goal is to avoid using Jordan's lemma, we need the following technical result.
\begin{lemma}\label{lem:supertech}
	Let $\Pi\in\field C^{d\times d}$ be a projector, and two sets of vectors $\{\ket{e_i} \}, \{\ket{f_i} \}$ which span $\field C^d$ such that
	\begin{align*}
	\braket{e_i}{e_j}=\braket{f_i}{f_j} & =\delta_{ij}       &  & (i)   \\
	\braket{e_i}{f_j}                   & =0                 &  & (ii)  \\
	\bra{e_i}\Pi\ket{e_j}               & \propto\delta_{ij} &  & (iii) \\
	\bra{f_i}\Pi\ket{f_j}               & \propto\delta_{ij} &  & (iv)
	\end{align*}
	Then either:
	\begin{enumerate}
		\item For all $i$, $\bra{e_i}\Pi\ket{f_j}\neq 0$ for at most one $j$.
		\item If $\bra{e_i}\Pi\ket{f_j}\neq 0$ and $\bra{e_i}\Pi\ket{f_{j'}}\neq 0$ for some $j'\neq j$, then necessarily $\bra{f_j}\Pi\ket{f_j}=\bra{f_{j'}}\Pi\ket{f_{j'}}$.
	\end{enumerate}
\end{lemma}
\begin{proof}
	First observe that $(i)$ and $(ii)$ imply that the set $\{\ket{e_i} \}\cup \{\ket{f_i} \}$ forms an orthonormal basis of $\field C^d$. For any $j$ for which
	\begin{equation}\label{eq:supertech-1}
	\bra{e_i}\Pi\ket{f_j}\neq 0,
	\end{equation}
	we can thus show
	\begin{align*}
	0 & \neq\bra{e_i}\Pi\ket{f_j}                                                                &  \\
	& =\bra{e_i}\Pi\Pi\ket{f_j}                                                                &  & \text{since $\Pi$ is a projector}                                      \\
	& =\bra{e_i}\Pi\left(\sum_k\ketbra{e_k} + \sum_l\ketbra{f_l} \right)\Pi\ket{f_j}           &  & \text{as the $\ket{e_i}$ and $\ket{f_i}$ form an ONB} \\
	& =\bra{e_i}\Pi\ket{e_i}\bra{e_i}\Pi\ket{f_j} + \bra{e_i}\Pi\ket{f_j}\bra{f_j}\Pi\ket{f_j} &  & \text{by $(iii)$ and $(iv)$}\\
	& =\bra{e_i}\Pi\ket{f_j}(\bra{e_i}\Pi\ket{e_i} + \bra{f_j}\Pi\ket{f_j}),                   &
	\end{align*}
	and thus $\bra{e_i}\Pi\ket{e_i} + \bra{f_j}\Pi\ket{f_j}=1$. If there exist two $j\neq j'$ for which \cref{eq:supertech-1} holds, then necessarily $\bra{f_j}\Pi\ket{f_j}=\bra{f_{j'}}\Pi\ket{f_{j'}}$.
\end{proof}

The following lemma now gives an explicit construction of the block-diagonalizing unitary $\op V$.
We will use the resulting block-diagonalizing unitary $\op V$ to bring $\Hkit$ into a shape similar to $\op A$, but where $\op V^\dagger\op B\op V$ connects the separate blocks with some off-diagonal entries.
By bounding the magnitude of the latter, we can finally lower-bound the overall spectrum of $\Hkit$ by $\Omega(T^{-2})$.
\begin{lemma}\label{lem:tech}
	Let $\op M\in\field C^{d\times d}$ be a projector. Then for any $1\le s< d$, there exists a block-diagonal unitary $\op V=\op V'\oplus\op V''$ with $\dim \op V'=s$ such that $\op D=\op V^\dagger\op M\op V$ is stoquastic. Furthermore, $\op V$ can be chosen such that---if we denote the rank of the upper-left $s\times s$ block with $r_a$ and the complementary lower-right block rank with $r_b$---$\op D_{ij}\neq 0$ if and only if
	\begin{align*}
	& i=j\land j\le r_a             && (i)\\
	\text{or}\quad & s\le i=j\le s+r_b             && (ii) \\
	\text{or}\quad & s\le j=s+i\le\min\{r_a,r_b\}  && (iii) \\
	\text{or}\quad & s\le i=s+j\le\min\{r_a,r_b\}. && (iv)
	\end{align*}
\end{lemma}
\begin{proof}
	Let $r=\rank\op M$, and write $\op M=\sum_{i=1}^r\ketbra{m_i}$ as eigendecomposition with orthonormal vectors $\ket{m_i}$.  Split each of the eigenvectors $\ket{m_i}=:\ket{a_i}\oplus\ket{b_i}$, where $\dim\ket{a_i}=s$ (we abuse the direct sum notation here to mean simply dovetailing two vectors).  Then $\op M$ can be expressed as block-matrix:
	$$
	\op M=\sum_{i=1}^r\ketbra{m_i}=:\sum_i\begin{pmatrix}
	\ketbra{a_i} & \ketbra{a_i}{b_i} \\
	\ketbra{b_i}{a_i} & \ketbra{b_i}
	\end{pmatrix}
	=:\begin{pmatrix}
	\op M_{aa} & \op M_{ab} \\
	\op M_{ab}^\dagger & \op M_{bb}
	\end{pmatrix}.
	$$
	The matrix $\op M_{aa}$ is Hermitian since $\op M$ is. We can thus find a diagonalizing unitary $\op A:=\sum_{i=1}^s\ketbra{\alpha_i}{i}$, such that $\op A^\dagger\op M_{aa}\op A$ is diagonal, and similarly $\op B$ for $\op M_{bb}$; here the $\ket i$ just label the standard basis, and the $\ket{\alpha_i}$ an orthonormal eigenbasis of $\op A$.  We further pick $\op A$ and $\op B$ such that the nonzero eigenvalues of each block are sorted to the top left, which we can always achieve.  If we now were to set $\op V=\op A\oplus\op B$, we would immediately verify claims $(i)$ and $(ii)$.
	
	It is, however, not obvious what $\op A^\dagger\op M_{ab}\op B$ looks like, so we need to do some more work. We calculate the matrix entries:		
	\begin{equation}\label{eq:tech-1}
	\bra x\op A^\dagger\op M_{ab}\op B\ket y = \sum_{j=1}^s\braket{x}{j}\bra{\alpha_j}\sum_i\ketbra{a_i}{b_i} \sum_{k=1}^{d-s}\ket{\beta_k}\braket{k}{y} = \sum_{i}\braket{\alpha_x}{a_i}\braket{b_i}{\beta_y}.
	\end{equation}
	We expand $\ket{\alpha'_i}:=\ket{\alpha_i}\oplus\ket0$, and $\ket{\beta'_i}:=\ket0\oplus\ket{\beta_i}$, which allows us to apply \cref{lem:supertech} to $\op M$, $\{\ket{\alpha'_i} \}$ and $\{\ket{\beta'_i} \}$. Take some row $x$; if there is only one column index $y$ for which \cref{eq:tech-1}$\neq0$ we are done.  If, on the other hand, there are at least two $y\neq y'$ for which \cref{eq:tech-1} $\neq 0$, we can define a unitary $\op R$ via
	\begin{align*}
	\ket{\beta'_y} &\longmapsto \frac{1}{\sqrt{1+\gamma^2}}\left(\ket{\beta'_y} + \gamma\ket{\beta'_{y'}}\right)\\
	\ket{\beta'_{y'}} &\longmapsto \frac{1}{\sqrt{1+\gamma^2}}\left(\ket{\beta'_y} - \gamma\ket{\beta'_{y'}}\right),    		
	\end{align*}
	where $\gamma=\bra{\alpha_x}\Pi\ket*{\beta'_y} / \bra{\alpha_x}\Pi\ket*{\beta'_{y'}}$.  It is then easy to verify that $\bra x(\op A\oplus\op B\op R)^\dagger\op M(\op A\oplus\op B\op R)\ket{y'}=0$, and because of the degeneracy $\bra y\op B^\dagger\op M_{bb}\op B\ket y=\bra {y'}\op B^\dagger\op M_{bb}\op B\ket{y'}$, $\op R^\dagger\op B^\dagger\op M\op B\op R$ is still diagonal.
	
	Iterate this process; once in every row there is at most one off-diagonal entry.  As each resulting matrix remains Hermitian, the same process immediately proves that there is at most one off-diagonal entry in every column.  It is now clear that we can conjugate this $\op V^\dagger\op M\op V$ with a diagonal matrix where each entry is of modulus 1, thus eliminating the different phases in the $\xi_i$, and such that the resulting sign is negative on all off-diagonal entries. Such a matrix is necessarily unitary and renders the resulting matrix stoquastic by definition, and we absorb this phase elimination unitary into $\op V$. 
	
	\noindent
	After some reordering, the resulting matrix $\op V^\dagger\op M\op V$ has the form
	\small
	\begin{equation}\label{eq:matrix-alt}
	\begin{tikzpicture}
	[baseline] \node at (0,0) {$\begin{pmatrix}
		\begin{array}{ccccccc@{\hskip 8mm}ccccccc}
		\lambda_1 &     0     &   0    &    \cdots     & \cdots &   0    &   0    & -|\xi_1| &    0     &   0    &    \cdots    & \cdots &   0    &   0    \\
		0     & \lambda_2 & \ddots &               &        &        &   0    &    0     & -|\xi_2| & \ddots &              &        &        &   0    \\
		0     &  \ddots   & \ddots &    \ddots     &        &        & \vdots &    0     &  \ddots  & \ddots &    \ddots    &        &        & \vdots \\
		\vdots   &           & \ddots & \lambda_{r_a} &   0    &        & \vdots &  \vdots  &          & \ddots & -|\xi_{r_o}| &   0    &        & \vdots \\
		\vdots   &           &        &       0       &   0    & \ddots & \vdots &  \vdots  &          &        &      0       &   0    & \ddots & \vdots \\
		0     &           &        &               & \ddots & \ddots &   0    &    0     &          &        &              & \ddots & \ddots &   0    \\
		0     &     0     & \cdots &    \cdots     & \cdots &   0    &   0    &    0     &    0     & \cdots &    \cdots    & \cdots &   0    &   0    \\[5mm]
		-|\xi_1|  &     0     &   0    &    \cdots     & \cdots &   0    &   0    &  \mu_1   &    0     &   0    &    \cdots    & \cdots &   0    &   0    \\
		0     & -|\xi_2|  & \ddots &               &        &        &   0    &    0     &  \mu_2   & \ddots &              &        &        &   0    \\
		0     &  \ddots   & \ddots &    \ddots     &        &        &        &    0     &  \ddots  & \ddots &    \ddots    &        &        & \vdots \\
		\vdots   &           & \ddots & -|\xi_{r_o}|  &   0    &        &        &  \vdots  &          & \ddots &  \mu_{r_b}   &   0    &        & \vdots \\
		\vdots   &           &        &       0       &   0    & \ddots &        &  \vdots  &          &        &      0       &   0    & \ddots & \vdots \\
		0     &           &        &               & \ddots & \ddots &   0    &    0     &          &        &              & \ddots & \ddots &   0    \\
		0     &     0     & \cdots &    \cdots     & \cdots &   0    &   0    &    0     &    0     & \cdots &    \cdots    & \cdots &   0    &   0
		\end{array}
		\end{pmatrix}$};
	\draw[dashed,gray] (6.6,0) -- (-6.6,0) (0,-4.3) -- (0,4.4);
	\end{tikzpicture}
	\end{equation}
	\normalsize
	where $r_o:=\min\{r_a,r_b\}$, which finally satisfies claims $(i)-(iv)$. The claim of the lemma follows.
\end{proof}
One can then show by an application of \cref{lem:fullrank} that for a \no-instance, $r_o=r=d/2$.

The remaining argument follows as in \cref{sec:jordan-proof}.

\end{document}